\documentclass[11pt]{article}
\usepackage[utf8]{inputenc} 
\usepackage[T1]{fontenc}    

\usepackage{xcolor}
\definecolor{ForestGreen}{rgb}{0.1333,0.5451,0.1333}
\definecolor{DarkRed}{rgb}{0.65,0,0}
\definecolor{Red}{rgb}{1,0,0}
\usepackage[linktocpage=true,
pagebackref=true,colorlinks,
linkcolor=DarkRed,citecolor=ForestGreen,
bookmarks,bookmarksopen,bookmarksnumbered, breaklinks]
{hyperref}      
\usepackage{url}            
\usepackage{booktabs}       
\usepackage{amsfonts}       
\usepackage{nicefrac}       
\usepackage{microtype}      

\usepackage[round,longnamesfirst]{natbib}

\usepackage{fullpage}
\usepackage[margin = 1in]{geometry}

\usepackage{amsmath, amssymb, amsthm}
\usepackage{thmtools}
\usepackage{cleveref}

\usepackage{mathtools}  
\usepackage{xfrac} 
\usepackage[normalem]{ulem}

\usepackage{algorithm}
\usepackage{algpseudocode}

\usepackage{graphicx}

\usepackage{tcolorbox}

\usepackage{tikz}
\usepackage{pgfplots}
\pgfplotsset{compat=1.18}

\newtheorem{theorem}{Theorem}[section]
\newtheorem{claim}[theorem]{Claim}

\newtheorem{lemma}[theorem]{Lemma}
\newtheorem{definition}[theorem]{Definition}

\newcommand {\Exp}       {\mathbb{E}}
\newcommand {\Prob}  [1] {\Pr \{#1 \}}

\newcommand {\E}     [1] {\Exp\left[#1\right]}

\newcommand {\ncoord}    {\ell}

\DeclareMathOperator {\cost}{cost}

\DeclareMathOperator{\OPT}{OPT}

\title{Socially Fair Clustering: Parameterized Approximation and Local Search}

\author{
Aditya Anand\thanks{University of Michigan, \texttt{adanand@umich.edu}} 
\and 
Yury Makarychev\thanks{Toyota Technological Institute at Chicago (TTIC), \texttt{yury@ttic.edu}}
\and 
Liren Shan\thanks{Toyota Technological Institute at Chicago (TTIC), \texttt{lirenshan@ttic.edu}}
}
\date{}

\begin{document}

\maketitle

\begin{abstract}
We study the Socially Fair Clustering problem introduced by~\cite{abbasi2021fair} and~\cite{ghadiri2021socially}, along with its extension, the $(p,q)$-Socially Fair Clustering problem. This problem generalizes $k$-medians and $k$-means to settings where data points are partitioned into $\ell$ groups, and the goal is to find a fair clustering that is simultaneously good for all groups. We present several algorithms for this problem.

\begin{enumerate}

\item For $\ell_p$-Socially Fair Clustering, we give the first constant-factor FPT-approximation parameterized by the number of groups $\ell$, resolving the open question raised by~\cite{ghadiri2022constant}. Our main ingredient is a new algorithm for closing additional centers in parameterized time inspired by local search.   

\item We then turn to the more general $(p,q)$-Socially Fair Clustering problem. The known algorithm for this problem, proposed by~\cite{chlamtavc2022approximating} achieves a very good approximation but is complex, slow and difficult to implement. We analyze the performance of a simple local search algorithm and show that it provides an $O(q)$ approximation in the worst case.

\item Finally, we design approximation algorithms for the facility location variant of the problem, where the number of facilities (centers) is not fixed in advance, and opening each facility incurs an opening cost. Unlike in previous work, we do not assume these opening costs are the same for all groups.
\end{enumerate}

\end{abstract}


\section{Introduction}
\subsection{Overview}
Today, computers are widely used to analyze and classify data, generate recommendations, and make decisions that affect many aspects of people’s lives. Traditionally, algorithms have focused on optimizing the overall quality of a solution. As a result, their decisions may be unfair to specific individuals or groups. Given the significance of these decisions, addressing fairness in algorithmic decision-making has become essential.

The $k$-medians and $k$-means objectives are widely used in clustering and typically yield solutions that perform well for most data points. However, some points may be much farther from their assigned centers than others, and therefore incur disproportionately high costs. This is especially problematic when data points represent people, and those bearing the highest costs disproportionately belong to minority groups. Consider, for example, a facility location problem modeled using $k$-medians: we wish to open $k$ facilities (e.g., hospitals or grocery stores) in a city to serve all residents effectively. The goal is to minimize the sum of distances from each person (data point) to their nearest facility (center): $\sum_{u \in X} d(u, C)$. If a minority group resides in an isolated neighborhood, it is possible that none of the facilities will be placed nearby, leaving that group underserved.

To address this issue, \cite{abbasi2021fair} and \cite{ghadiri2021socially} independently introduced the \textit{Socially Fair Clustering} problem. This problem was also previously studied by~\cite{anthony} in the context of clustering under uncertainty.
In the fair clustering setting, the $n$ data points belong to $\ell$ groups $G_1, \dots, G_\ell\subseteq X$, representing minority or otherwise protected populations (the groups may overlap and do not have to cover $X$). The goal is to find a clustering that performs well for all groups simultaneously. 
It is convenient to assign weights to data points: for group $i$, let $w_i(u)$ denote the weight of point $u$, where $w_i(u) > 0$ if $u \in G_i$ and $w_i(u) = 0$ otherwise; one natural choice is $w_i(u) = 1/|G_i|$ for $u \in G_i$.

The $\ell_p$-Socially Fair Clustering problem, referred to as $\ell_p$-Clustering for brevity, asks for a set of $k$ centers $C$ from a set of potential centers $F \subseteq X$ that minimizes
\begin{equation} \label{eq:ell-p-objective}
\max_{i \in \{1,\dots,\ell\}} \left( \sum_{u \in X} w_i(u) d(u, C)^p \right)^{1/p},
\end{equation}
ensuring that the solution performs well across all groups. In other words, we compute the $p$-norm of distances $\left( \sum_{u \in X} w_i(u) d(u, C)^p \right)^{1/p}$ for each group $G_i$, and the objective aggregates these by taking a maximum (many papers define the objective as the $p$-th power of ours, the two formulations are essentially equivalent: an $\alpha$-approximation for our objective corresponds to an $\alpha^p$-approximation for the alternative, and vice versa).

A more general variant of the problem, known as $(p,q)$-Socially Fair Clustering with $1 \leq p \leq q \leq \infty$, employs a different aggregation function.\footnote{We assume that $q \geq p$, since otherwise the objective would favor \textit{unfair} clusterings. 
For example, let $p = 2q$ and $\ell = 2$. Suppose there are two candidate solutions: in the first, each group has $p$-norm distance $1$, while in the second, the $p$-norm of the distances for the two groups are $2^{1/p}$ and $0$ (where the $p$-norm distance of group $i$ is given by $\left( \sum_{u \in X} w_i(u) d(u, C)^p \right)^{1/p}$). The corresponding aggregated objectives are $(1^{1/2} + 1^{1/2})^{1/q} = 2^{1/q}$ and $(2^{1/2} + 0^{1/2})^{1/q} = 2^{1/(2q)}$. 
Thus, we would prefer the second solution, even though it is clearly less fair.}
In this formulation, we construct a vector whose $i$-th coordinate represents the $p$-norm of the distances for group $i$, and then take the $\ell_q$ norm of this vector:
\begin{equation}\label{eq:pq-objective}
\left( \sum_{i=1}^{\ell} \left( \sum_{u \in X} w_i(u) d(u, C)^p \right)^{q/p} \right)^{1/q}.
\end{equation}
We refer to this problem simply as $(p,q)$-Clustering. When $q = \infty$, the objective in~\eqref{eq:pq-objective} reduces to that in~\eqref{eq:ell-p-objective}; thus, $\ell_p$-Clustering and $(p,\infty)$-Clustering are equivalent.
Moreover, the $(p,\log \ell)$-objective approximates the $\ell_p$-objective within a constant factor. This follows from the more general fact that the $\ell_q$ and $\ell_r$ norms of a $d$-dimensional vector are within a factor of $d^{|\frac{1}{q} - \frac{1}{r}|}$ of each other, which can be shown by a simple computation. Hence in all our results for $(p,q)$-Clustering, we obtain results for $\ell_p$-Clustering by setting $q = \log \ell$.

\subsection{Known results} The problem has been extensively studied since its introduction. We begin by summarizing previously published approximation algorithms.
\cite{anthony} gave an $O(\log n + \log \ell)$-approximation for $\ell_1$-Clustering.
\cite{abbasi2021fair} and \cite{ghadiri2021socially} provided $O(\ell^{1/p})$-approximation algorithms for $\ell_p$-Clustering with $p \in \{1, 2\}$. Additionally, \cite{abbasi2021fair} presented a bicriteria approximation algorithm that achieves a constant-factor approximation by opening $O(k)$ centers.

The approximation factor was later improved by \cite{MV21}, who obtained an $O\left(\frac{\log \ell}{\log \log \ell}\right)^{1/p}$-approximation. 
This approximation factor is optimal under the assumption that $NP\not\subseteq \bigcap_{\delta> 0}DTIME(2^{n^\delta})$, as was shown by \cite{Bhattacharya}.
Finally, \cite{chlamtavc2022approximating} studied the $(p,q)$-Clustering problem and gave an $O\left(\frac{q}{\ln(1 + q/p)}\right)^{1/p}$-approximation for it.
This approximation factor matches that for $\ell_p$-clustering when $q=\log \ell$.

Another line of research has focused on designing fixed-parameter tractable (FPT) algorithms for $\ell_p$-Clustering.
In a significant result, \cite{goyal2023tight} settled the approximability of $\ell_p$-Fair Clustering for FPT algorithms parameterized by $k$ by designing a $3 + \varepsilon$-approximation algorithm and proving a matching lower bound of $3 - \varepsilon$.
They also established lower bounds of $1 + 2/e - \varepsilon$ and $\sqrt{1 + 8/e} - \varepsilon$ for the approximability of $\ell_1$ and $\ell_2$-Clustering when parameterized by both $k$ and $\ell$.

\cite{ghadiri2022constant} observed that in practice, $\ell$ is typically a small constant, while $k$ may be very large.
They presented a polynomial-time algorithm achieving a $5 + 2\sqrt{6}$ bicriteria approximation by opening at most $k + \ell$ centers. Building on the bicriteria result mentioned above, they also designed algorithms with running times $n^{O_p(\ell^2)}$ and $k^{\ell} \mathrm{poly}(n)$, achieving approximation factors of $5 + 2\sqrt{6} + \varepsilon$ and $15 + 6\sqrt{6} + \varepsilon$, respectively.
The latter is an FPT algorithm parameterized by both $k$ and $\ell$. They then raised the question of whether there exists a FPT algorithm achieving a constant approximation parameterized solely by $\ell$.
This question remained open until our current work.

We also note that the problem admits a $(1 + \varepsilon)$-approximation when parameterized by $k$ for certain special classes of metric spaces, as shown by~\cite{abbbasi-EPAS}.
In many real-world applications, the number of centers—typically referred to as facilities in this context—is not fixed in advance. Instead, opening each facility incurs a cost. As is standard, we refer to this setting as the Facility Location problem. \cite{abbasi2021fair} presented a 4-approximation algorithm for the Facility Location variant of $\ell_1$-Clustering, assuming homogeneous costs across groups (i.e., the cost of opening a facility is shared equally among all groups).

\subsection{Other related work}
The theme of fair clustering has received a lot of attention in recent years. Some other notions of fairness that have been adopted include balancedness~\cite{chierichetti2017fair}, no over-representation~\cite{ahmadian2019clustering}, fair hierarchical clustering~\cite{ahmadian2020fair}, individual fairness~\cite{mahabadi2020individual}, pairwise fairness~\cite{bandyapadhyay2024polynomial, bandyapadhyay2025constant}. Many of these incorporate constraints on clusters rather than enforcing fairness through a new objective, as is the case for Socially Fair Clustering.

\subsection{Our results}
\subsubsection{\texorpdfstring{FPT Algorithm parameterized by $\ell$}{FPT Algorithm parameterized by l}}
In this paper, we address several important open problems in the field. First, we resolve the open question posed by~\cite{ghadiri2022constant} by presenting a constant-factor approximation FPT algorithm for $\ell_p$-Clustering, parameterized by the number of groups $\ell$.

\begin{theorem}\label{thm:FPTmain}
For every $\varepsilon > 0$, there exists a $(17 + 6\sqrt{6} + \varepsilon)$-approximation algorithm for $\ell_p$-Clustering that runs in time $f(\ell, p, \frac{1}{\varepsilon}) \cdot n^{O(1)}$.
\end{theorem}

We prove this result by designing a new algorithm that reduces the number of centers in fixed-parameter time.

\begin{theorem}\label{thm:FPT}
There exists an algorithm that, given an instance of $\ell_p$-Clustering, a $c$-approximate solution $C$ with $k + t$ centers, and a parameter $\varepsilon > 0$, returns a $(3c + 2 + \varepsilon)$-approximate solution with $k$ centers. The algorithm runs in time $f(\ell, p, t, \frac{1}{\varepsilon}) \cdot n^{O(1)}$.
\end{theorem}

Combining~\Cref{thm:FPT} with the following result of~\cite{ghadiri2022constant}, we immediately obtain~\Cref{thm:FPTmain}, thereby answering their open question.

\begin{theorem}[\cite{ghadiri2022constant}]\label{thm:ghadiri-et-al}
There is a bicriteria polynomial-time algorithm for $\ell_p$-Clustering that finds a $(5 + 2\sqrt{6})$-approximate solution with $k + \ell$ centers.
\end{theorem}

Additionally, we strengthen Theorem~\ref{thm:ghadiri-et-al} in the important case of $\ell_1$-Clustering by presenting a local-search algorithm that finds a $(3+ \varepsilon)$-approximation algorithm with $k+f(\ell,\frac{1}{\epsilon})$ centers. 
(\Cref{thm:localsearchconstant}). 
This yields a variant of the algorithm from Theorem~\ref{thm:FPTmain} with a $(11+\varepsilon)$-approximation for $\ell_1$-Clustering. 

\subsubsection{Local Search Algorithm for \texorpdfstring{$(p,q)$}{(p,q)}-Clustering}
The second theme of our paper is the design of a more efficient and implementable algorithm for $(p,q)$-Clustering. The algorithm in~\cite{chlamtavc2022approximating} solves an exponential-size convex relaxation using the round-or-cut framework, which is impractical and difficult to implement. In contrast, local search algorithms are well known for their efficiency and ease of implementation.
Previously, \cite{Arya} showed how to obtain a $(5 + \varepsilon)$-approximation for $k$-median using a single-swap local search algorithm, and a $(3 + \varepsilon)$-approximation using a multiple-swap variant. This result was generalized by \cite{GT}, who gave an $O(p)$-approximation for the single-group case of $\ell_p$-Clustering and showed that this analysis is tight -- there are instances where the algorithm returns only an $\Omega(p)$-approximate solution.

We analyze the performance of the single-swap local search algorithm on instances of $(p,q)$-Clustering with an arbitrary number of groups $\ell$, and show that it yields an $O(q)$-approximation. The aforementioned lower bound in~\cite{GT} implies that our analysis is tight. While our approximation guarantee is worse than those of~\cite{MV21} and~\cite{chlamtavc2022approximating}, it comes close in a particularly important case: $\ell_1$-Clustering. For that problem, we obtain an $O(\log \ell)$-approximation by setting $q = \log \ell$, which matches the best known approximation factor of $O(\log \ell / \log \log \ell)$ up to a $\log \log \ell$ factor. Our algorithm yields a simple, implementable and much faster algorithm compared to the algorithm of~\cite{chlamtavc2022approximating}. \footnote{The running time in~\cite{chlamtavc2022approximating} is polynomial, but it is not stated explicitly.}

\begin{theorem}\label{thm:local_search}
There is a single-swap local search algorithm for the $(p,q)$-Clustering problem that gives an $O(q)$ approximation, performs at most $O(kq \log n)$ iterations and runs in time $O_{p,q}(n^2\ell k^2 \log n)$.
\end{theorem}

\subsubsection{\texorpdfstring{$(p,q)$}{(p,q)}-Socially Fair Facility Location}
Finally, we study the $(p,q)$-Socially Fair Facility Location problem, an important and practically motivated variant of $(p,q)$-Clustering.
Building on the result of~\cite{abbasi2021fair} for $\ell_1$, we present a $(4 + \varepsilon)$-approximation algorithm for $(p,q)$-Socially Fair Facility Location with homogeneous costs.
We then consider a more general setting with heterogeneous costs, where different groups incur different facility opening costs. In this case, we obtain an $O(q / \log q)$-approximation; in particular, for $p = 1, q = \infty$, this yields an $O(\log \ell / \log \log \ell)$-approximation\footnote{We recall that the $(1, \log \ell)$ objective approximates the $(1, \infty)$ objective upto a constant factor}, which almost matches the approximation hardness of $\Omega(\log^{1-\varepsilon} \ell)$ for this problem; 
see Theorem~\ref{thm:fl_vector_hardness} for details.

\begin{restatable}{theorem}{FLscalar}
\label{thm:fl_scalar}
For Socially Fair Facility Location with group homogeneous facility cost, there exists a polynomial-time algorithm that achieves a $4+\varepsilon$ approximation for any $\varepsilon > 0$.
\end{restatable}

\begin{restatable}{theorem}{FLvector}
\label{thm:fl_vector}
For Socially Fair Facility Location with group heterogeneous facility cost, there exists a polynomial-time algorithm that achieves an $O(q/\log q)$ approximation.
\end{restatable}

\section*{Organization}
\Cref{sec:overview} gives an overview of the proofs of our results. In Section~\ref{sec:fpt}, we present our fixed-parameter tractable (FPT) algorithm for $\ell_p$-Clustering. In Section~\ref{sec:local-search}, we analyze a simple single-swap local search algorithm for $(p,q)$-Clustering. In Section~\ref{sec:facility}, we describe our algorithm for $(p,q)$-Socially Fair Facility Location. 
In Appendix~\ref{sec:extra-centers}, 
we give a $(3+\varepsilon)$-bicriteria local-search algorithm that opens $k + f(\ell)$ centers for $\ell_1$-Clustering.
\section{Preliminaries}
\label{sec:prelim}
This paper focuses on the $\ell_p$ and $(p,q)$-Fair Clustering Problems, referred to below simply as $\ell_p$- and $(p,q)$-Clustering. In these problems, we are given a metric space $(X,d)$, where the points represent data points or clients, and a set of potential center or facility locations $F \subseteq X$. Let $n$ be the number of data points or clients. For each point $u \in X$, we are also given a vector of weights $w(u) \in \mathbb{R}_{\geq 0}^\ell$, where the $i$-th coordinate $w_i(u)$ denotes the weight for group $i$. We denote the number of groups by $\ell$. The goal is to select a subset of $k$ centers $C \subseteq F$ that minimizes objective~\eqref{eq:ell-p-objective} for $\ell_p$-Clustering and~\eqref{eq:pq-objective} for $(p,q)$-Clustering.

To analyze these objectives, we define the $\ell_p$-cost of the clustering (raised to the power $p$) with respect to a weight function $w_i: X \to \mathbb{R}_{\geq 0}$ and centers $C$ as 
    $$
    \cost(w_i,C) = \sum_{u \in X} w_i(u) d(u,C)^p,
    $$
where $d(u,C) = \min_{c \in C} d(u,c)$. 
We then define the cost vector induced by weights $w$ and centers $C$ as the $\ell$-dimensional vector whose $i$-th coordinate is $\cost(w_i,C)^{1/p}$.
Then the $\ell_p$-objective in~\eqref{eq:ell-p-objective} is the $\ell_\infty$-norm of the cost vector, while the $(p,q)$-objective in~\eqref{eq:pq-objective} is its $\ell_q$-norm.

We will need the following claim.
\begin{claim}[see e.g.~Lemma A.1 in  \cite{MMR}]\label{claim:MMR}
    For every non-negative real numbers $a$ and $b$, and for every $p \geq 1$ and $\varepsilon > 0$, we have
    \[(a+b)^p \leq (1+\varepsilon)^{p-1} a^p + (1+1/\varepsilon)^{p-1} b^p,\]
\end{claim}

\section{Proof Overview}\label{sec:overview}
In this section, we offer a technical overview of the proofs of our main results. We focus primarily on the fixed-parameter tractable (FPT) approximation algorithm for $\ell_p$-Clustering parameterized by the number of groups $\ell$ (Section~\ref{sec:fpt}), and then briefly describe the ideas underlying the remaining results.

\subsection{FPT approximation for \texorpdfstring{$\ell_p$}{lp}-Clustering (Section~\ref{sec:fpt})}
Our goal is to design a constant-factor approximation algorithm for $\ell_p$-Clustering that runs in time $f(\ell)\cdot n^{O(1)}$, parameterized only by the number of groups $\ell$. As a starting point, we run the algorithm of \cite{ghadiri2022constant} (see Theorem~\ref{thm:ghadiri-et-al}) to get a $(5+2\sqrt{6})$-approximate solution with at most $k+\ell$ centers in polynomial time. 
In the special case of $\ell_1$-Clustering, we can alternatively use the local-search-based bicriteria approximation algorithm 
from Section~\ref{sec:extra-centers} (see Theorem~\ref{thm:localsearchconstant}), 
which yields a better approximation. 
Let $C$ denote the resulting set of $k+t$ centers, where $t=\ell$, and let $c$ be the corresponding approximation factor.

Given such a bicriteria solution, the main challenge is to close $t$ centers and reassign the clients assigned to them to other centers, while controlling the cost increase for all groups. We proceed in two steps. First, we show that \emph{existentially} it is possible to close $t$ centers $C_1$ and reassign every client assigned to a center $c$ in $C_1$ to the center closest to $c$ in $C\setminus C_1$ so that the cost of the resulting solution is within a constant factor of the cost of $C$ (see Lemma~\ref{lemma:redirect} for the exact guarantee). The key idea is that the new connection cost for a client $u$ can be bounded in terms of both the cost of connecting $u$ to its originally assigned center in $C_1$ and the cost of connecting $u$ to the center serving $u$ in the optimal solution $C^*$. This observation is important because it allows us to reassign the clients assigned to the $t$ closed centers to at most $t$ centers (that are not closed), and this property is crucial for our algorithm.

However, we cannot use the procedure from Lemma~\ref{lemma:redirect} directly, as it requires knowledge of $C^*$. Thus, to obtain an FPT algorithm, we need to identify the set $C_1$ without access to the optimal solution. This is achieved using the technique of color coding introduced by~\cite{alon1995color}.
Let $C_2 \subseteq C \setminus C_1$ be the set formed by selecting, for each $c \in C_1$, a closest center to $c$ in $C\setminus C_1$.
Note that $|C_2| \leq |C_1| = t$ (each center $c_1\in C_1$ contributes at most one selected center; and different centers in $C_1$ may potentially have the same closest center in $C\setminus C_1$). We color all centers in $C$ randomly using two colors, red and blue. Each center is colored red or blue independently with probability $1/2$. Note that with probability at least $2^{-2t}$, all centers in $C_1$ are colored red and all centers in $C_2$ are colored blue. By running the algorithm $2^{\Omega(t)}$ times, we may assume that we obtain a coloring with this property. We denote the sets of red and blue centers by $R$ and $B$, respectively.

Now, we introduce the notion of a $Z$-profile of a center $c$. Loosely speaking, the $Z$-profile is a vector in $\mathbb{R}^\ell$ that captures, for each of the $\ell$ groups, the reassignment cost incurred by closing center $c$ and reassigning all clients assigned to $c$ to the center closest to $c$ in $Z$.
Crucially, the $C_2$-profiles of all centers in $C_1$ fully determine the reassignment cost: the reassignment cost vector is simply the sum of the $C_2$-profiles of all centers in $C_1$. Importantly, if we were to find another set of centers in $R$ with the same profiles, the reassignment cost would be the same. However, there are two challenges in finding such a set. First, as defined, the number of possible profiles is unbounded. Second, since we do not know the set $C_2$, we cannot compute the $C_2$-profiles of centers in $R$.

To resolve the first issue, we discretize the set of all possible profiles so that their number is bounded by a function of $\ell$, $t$, $p$, and $1/\varepsilon$. We then ``guess'' the number of centers in $C_1$ with each possible profile. Next, we observe that the $C_2$-profiles and $B$-profiles of centers in $C_1$ are equal, and thus we can use $B$ as a proxy for $C_2$. We select centers in $R$ according to our guess of the profile counts and obtain a set $D_1$ of $t$ centers to close. Given the discussion above, it is not hard to see that by closing the centers in $D_1$ and reassigning clients to their closest centers in $B$, we obtain a constant-factor approximation.

\subsection{Local search for \texorpdfstring{$(p,q)$}{(p,q)}-Clustering (Section~\ref{sec:local-search})}
Building on the local search algorithms from~\cite{Arya} and~\cite{GT} for $k$-medians, $k$-means, and $(k,p)$-clustering, we design and analyze a local search algorithm for the more general $(p,q)$-Clustering problem.
The general approach is similar to that of~\cite{Arya} and~\cite{GT}, but the analysis is more complex, as it requires fine-grained control over the distribution of swaps (see below).

The algorithm starts with an arbitrary set $C$ of $k$ centers and repeatedly finds and performs a swap $c \to c'$ that replaces one of the centers $c \in C$ with a center $c'$ outside of $C$. The swaps are chosen to minimize the $(p,q)$-Clustering objective, and the algorithm terminates when no improving swap exists.
To prove that the algorithm achieves an $O(q)$-approximation, we show that if the cost of the current solution exceeds the optimal cost by $\Omega(q)$ times, then there must exist an improving swap.
To this end, we present a distribution over special swaps that replace centers in $C$ with centers in an unknown optimal solution $C^*$ (The algorithm may choose a swap from the distribution or any other swap improving the objective; we use the distribution only for the analysis). 
While this approach resembles that of~\cite{Arya} and~\cite{GT}, we require significantly more control over the distribution of swaps to handle the more complex objective. In particular, we use careful linear approximations for the change in the non-linear objective that allow us to show that there is indeed such a distribution over improving swaps.

To obtain these linear approximations, roughly speaking, we can obtain upper and lower bounds on the derivative of the cost function in the desired range. One can then use the mean-value theorem to find a linear approximation of the change in the objective using these bounds. The actual proof is slightly different, and uses convexity arguments along with the mean value theorem. Carefully combining this with the standard linear-objective analysis for $k$-Median, we show that the algorithm is an $O(q)$ approximation.

\subsection{Socially Fair Facility Location (Section~\ref{sec:facility})}
Finally, we study $(p,q)$-Socially Fair Facility Location. We consider two variants of the problem: facility location with homogeneous opening costs and facility location with heterogeneous opening costs. For the homogeneous case, we use a rounding algorithm inspired by the standard facility location algorithm of~\cite{shmoys1997approximation}. For the heterogeneous case, we replace the greedy step that selects and opens the cheapest facility within a certain ball $B_u$ with a randomized rounding step that opens a facility in $B_u$ with probability proportional to the target probability given by the convex relaxation. We then use Rosenthal's inequality to bound the expected $\ell_q$-norm of the solution. We also show a hardness result that shows that our algorithm is near-optimal for $\ell_1$-Socially Fair Facility Location (the case when $p = 1$, $q = \infty$).

\section{FPT algorithm for \texorpdfstring{$\ell_p$}{lp}-socially fair clustering}
\label{sec:fpt}
As discussed in the introduction, we provide an FPT algorithm for $\ell_p$-Clustering, parameterized by the number of groups $\ell$. 
The algorithm first finds an approximate solution with $k+t$ centers using Theorem~\ref{thm:ghadiri-et-al}.
This solution provides only a \textit{bicriteria} approximation. To obtain a \textit{true} approximation, the algorithm closes $t$ centers -- using Theorem~\ref{thm:FPT}.

We now proceed with the proof of~\Cref{thm:FPT}. Given the instance of $\ell_p$-Clustering on point set $X$ and metric $d$, a set of centers $C$ and an assignment $\alpha: X \rightarrow C$, we define the cost with respect to the weight function $w_i$ and assignment $\alpha$ as follows:

\begin{align*}
     \cost_i(\alpha) = \cost(w_i,\alpha) &:= \sum_{u \in X} w_i(u) d(u,\alpha(u))^p.
\end{align*}

Note that the $\ell_p$-objective function for the assignment $\alpha$ is then $\max_{i \in [\ell]} \cost_i(\alpha)^{1/p}$. Throughout this section, we fix an optimal set of centers $C^*$ and a corresponding optimal assignment $\beta^*$. Let us define $\OPT$ to be the optimal assignment $\ell_p$-objective, so that $\OPT = \max_i \cost_i(\beta^*)^{1/p}$.

First, we show that given any set $C$ of $k + t$ centers which is  a $c$-approximate solution, indeed there exists a subset of $t$ centers that can be closed, such that the resulting subset of $k$-centers is a $(3c + 2)$-approximate solution. However, besides this, we show a stronger property that our algorithm requires: every client assigned to some center of $C$ (in an optimal assignment of clients to $C$) that is closed in this process can now be re-assigned to one among at most $t$ many open centers in $C$. The next key lemma proves this result. Before we state the lemma, the following definition will be useful.

\begin{definition}
Given a set of open centers $C$, an assignment $\alpha: X \rightarrow C$ of points to centers in $C$, and two disjoint subsets $C_1, C_2 \subseteq C$, we define a new assignment $\alpha'$ obtained from $\alpha$ by reassigning clients as follows. For each client $u \in X$, if $\alpha(u) \in C_1$, we find the closest center $c_2$ in $C_2$ to $c_1 \equiv \alpha(u)$ and let $\alpha'(u) = c_2$. If $\alpha(u) \notin C_1$, we set $\alpha'(u) = \alpha(u)$ (do not change the assignment).

We then define $\cost(\alpha, C_1
, C_2)$ to be the vector $(v_1, v_2 \ldots v_{\ell})$ with $v_i = \cost_i(\alpha') - \cost_i(\alpha)$ for each $i \in [\ell]$.
\end{definition}

Simply put, $\cost_i(\alpha, C_1, C_2)$ measures the increase in cost for group (coordinate) $i$ by re-assigning each point assigned to a center in $c_1 \in C_1$ to the center in $C_2$ closest to $c_1$.
We use the shorthand $\cost(\beta, c_1, C_2)$ to denote $\cost(\beta, \{c_1\}, C_2)$. If $\beta$ is an optimal assignment, then $\cost(\beta, c_1, C_2) \geq 0$, since any reassignment can only increase the cost. We also have,
\begin{equation}\label{eq:additivity}
\cost(\beta, C_1, C_2) = \sum_{c_1\in C_1}  \cost(\beta, c_1, C_2)
\end{equation}

We formalize this nearest-neighbor redirection operation in the following lemma, showing that one can close a set $C_1$ of $t$ centers and reassign all affected clients to a set $C_2$ of at most $t$ receiving centers while keeping the reassignment cost controlled. 
The redirection operator in the definition captures exactly what happens when we close centers $C_1\subseteq C$: for each closed center $c_1\in C_1$, we send all of its clients to the closest open center $c_2\in C\setminus C_1$. Let $C_2$ be the set of such receiving centers. 

\begin{lemma}\label{lemma:redirect}
Given a $c$-approximate solution $C$ with $k + t$ centers and a corresponding optimal assignment $\beta$ of clients to centers in $C$, there exist disjoint sets $C_1,C_2 \subseteq C$ with $|C_2| \leq |C_1| = t$ such that 
\begin{itemize}
    \item $\cost(\beta, C_1, C_2) \leq ((3c + 2)^p - c^p)\mathrm{OPT}^p$
    \item $C_2 = \{c_2: c_2 $ is the closest point from $C\setminus C_1$ to $c_1$ for some $c_1 \in C_1\}$.
\end{itemize}
\end{lemma}
\begin{proof}

Let $C^* = \{c_1^*, c_2^*, \ldots, c_k^*\}$ be an optimal set of centers, and let $\beta^*$ be an optimal assignment of clients to the centers in $C^*$. For each $c_i^*$, mark $\gamma(c_i^*) \in C$, the closest center to $c_i^*$ in $C$. 
Then there are at most $k$ marked centers in $C$ and at least $t$ unmarked ones. Choose any subset of $t$ unmarked centers in $C$ to be the set $C_1$.
For each center of $C_1$, we take the closest center in $C \setminus C_1$ to obtain the set $C_2$. Note that this means $|C_2| \leq |C_1|$.

We now show that this choice of $C_1, C_2$ satisfies the statement of the lemma. Let $c_1 \in C_1$ and let $c_2$ be the closest center to $c_1$ in $C \setminus C_1$ (which must be in $C_2$). Then, for every point $u \in X$ with $\beta(u) = c_1$, we reassign $u$ to $c_2$.

We wish to bound the reassignment cost of $u$ to $c_2$. Let $c^* = \beta^*(u)$ be $u$'s assigned center in $C^*$.
By the triangle inequality, 
\begin{align*}
d(u,c_2) 
  &\le d(u,c_1) + d(c_1,c_2) \\
  &\le d(u,c_1) + d\bigl(c_1,{\gamma}(c^*)\bigr)\, \text{($c_2$ is the closest center to $c_1$ in $C \setminus C_1$)}\\
  &\le d(u,c_1) + d(c_1,c^*) + d\bigl(c^*,{\gamma}(c^*)\bigr) 
  \\
  &\le d(u,c_1) + 2d(c_1,c^*)\, \text{(by definition of $\gamma(c^*)$)} \\
  &\le d(u,c_1) + 2\bigl(d(u,c_1) + d(u,c^*)\bigr)\,\text{(again by triangle inequality)}\\ 
  &\le 3d(u,c_1) + 2d(u,c^*)
\end{align*}

It follows from Claim~\ref{claim:MMR} that for every $\delta > 0$, we get the inequality
\[
d(u,c_2)^p \leq 3^p d(u,c_1)^p (1 + \delta)^{p-1} + 2^p d(u,c^*)^p \left(1 + \frac{1}{\delta}\right)^{p-1}
\]

We can now use this to bound the change in cost. By definition, it follows that the contribution of client $u$ to $\cost_i(\beta,C_1,C_2)$ is  $w_i(u)\bigl(d(u,c_2)^p - d(u,c_1)^p\bigr)$. 
Now repeating the analysis for each point $u$, and recalling that $c_1 = \beta(u), c^* = \beta^*(u)$, we obtain
\begin{align*}
\cost_i(\beta, C_1, C_2)
\le& \sum_{u \in X}w_i(u)\Bigl(3^p (1 + \delta)^{p-1}\,d\bigl(u,\beta(u)\bigr)^p
   \\
   &+ 2^p\Bigl(1 + \frac{1}{\delta}\Bigr)^{p-1} d\bigl(u,\beta^*(u)\bigr)^p- d(u,\beta(u))^p\Bigr)\\
\le& \Bigl((3c)^p(1 + \delta)^{p-1} + 2^{p}\Bigl(1 + \frac{1}{\delta}\Bigr)^{p-1} - c^p\Bigr)\,\OPT^p
   \quad\text{for each }i \in [\ell],
\end{align*}
where we use the fact that $\sum_{u \in X} w_i(u)d\bigl(u,\beta^*(u)\bigr)^p \leq \OPT^p$ and $\sum_{u \in X} w_i(u) d\bigl(u,\beta(u)\bigr)^p \leq c^p\OPT^p$, since $C$ is a $c$-approximate solution.

Now we set \(\delta = \frac{2}{3c}\). It then follows that
\[
\cost_i(\beta, C_1, C_2) \le ((3c + 2)^p - c^p)\OPT^p.
\]
\end{proof}

Next, we describe how to find such sets in FPT time. We will use the color coding technique of \cite{alon1995color}. We remark that while our algorithm is stated as a randomized algorithm for simplicity, it can be easily derandomized using standard techniques (see for example.~\cite[Lemma~1]{chitnis2016designing}).

\begin{lemma}\label{lemma:findsets}
Given a c-approximate solution $C$ with $k + t$ centers and a corresponding optimal assignment $\beta$ of clients to centers in $C$, one can find disjoint sets $D_1, D_2$ with $|D_2| \leq |D_1| = t$ satisfying $\|\cost(\beta, D_1, D_2)\|_\infty \leq \bigl((3c + 2 )^p - c^p + \varepsilon^p\bigr) \OPT^p$ in time $f(\ell, t, p, \frac{1}{\varepsilon})n^{O(1)}$.
\end{lemma}

\begin{proof}

Let $C_1, C_2 \subseteq C$ be the (unknown) subsets of $C$ guaranteed by~\Cref{lemma:redirect}. 
We use color coding. Randomly and independently color each center in $C$ red or blue, each with probability half. Call a coloring \emph{good} if every center of $C_1$ is colored red and every center of $C_2$ is colored blue. Then a random coloring is good with probability at least $2^{-2t}$. By independently sampling $O(2^{2t}\log n)$ random colorings, we can ensure that we enumerate a good coloring with high probability. Henceforth, assume that we are given a good coloring (we will try all the $O(2^{2t} \log n)$ colorings). Let $R \subseteq C$ be the red colored centers, and $B \subseteq C$ be the blue colored centers in this good coloring.

 For a set $Z \subseteq C$, we say that the $Z$-profile of $c \in C$ is $(i_1, i_2, \ldots, i_{\ell})$ if the $m^{th}$ coordinate of $\cost(\beta, c, Z)$ is in the interval $\left[\frac{(i_{m} - 1)\varepsilon^p}{t}\OPT^p, \frac{i_m\varepsilon^p}{t}\OPT^p\right]$ for each $m \in [\ell]$. Here $i_1, i_2 \ldots i_m$ are positive integers. Let $\mathcal{P}$ be the set of profiles with all $i_j \leq t(\frac{3c + 2}{\varepsilon})^p$. Then $|\mathcal{P}| \leq ({t}(\frac{3c + 2}{{\varepsilon}})^p)^{\ell}$. 
 Note that the $C_2$-profile of every $c_1\in C_1$ lies in $\cal P$. This is the case because 
 (1) $\cost(\beta, c_1, C_2) \geq 0$, since $\beta$ is an optimal assignment and 
 (2) $\cost(\beta, c_1, C_2) \leq \cost(\beta, C_1, C_2) \leq (3c +2)^p\OPT^p$ by~\eqref{eq:additivity} and~\Cref{lemma:redirect}.

For each profile $P \in \mathcal{P}$, guess the number of centers in $C_1$ with $C_2$-profile $P$. Let $n_P$ denote this number for each $P$. Then $\sum_{P \in \mathcal{P}} n_P = |C_1| = t$. The number of guesses we need to make is at most $|\mathcal{P}|^t \leq \bigl({t}(\frac{3c +2}{\varepsilon})^p\bigr)^{\ell t}$.

Consider $c_1\in C_1$ and let $c_2$ be the closest point to $c_1$ in $C\setminus C_1$. Note that $c_2\in C_2$ (per item 2 in Lemma~\ref{lemma:redirect}) and $c_2$ also is the closest point to $c_1$ in $B$, since $c_2 \in C_2 \subseteq B\subseteq C\setminus C_1$.
Therefore, the $C_2$- and $B$-profiles of $c_1$ are the same. In particular, there are $n_P$ centers with $B$-profile $P$ in $C_1$. Since all centers in $C_1$ are red (recall that we assumed that the coloring is good), there are at least $n_P$ red centers with $B$-profile $P$.

Now for each $P \in \mathcal{P}$, we choose exactly $n_P$ red centers whose $B$-profile is $P$ to form the set $D_1$\footnote{while we do not know $OPT$, we can guess $OPT^p$ upto a factor of $(1 + \epsilon)$, but we omit this detail in the proof for the sake of simplicity}. 
Once we choose $D_1$, we let $D_2$ to be the set of all centers $d_2 \in B$ such that $d_2$ is the closest center in $B$ to some $d_1 \in D_1$. This concludes the construction of the sets $D_1$ and $D_2$.

\begin{lemma}\label{lemma:costfpt}
$\|\cost(\beta, D_1, D_2)\|_{\infty} \leq ((3c + 2)^p - c^p +\varepsilon^p)\OPT^p$.
\end{lemma}

\begin{proof}
There is a one-to-one map $f: C_1 \rightarrow D_1$ such that the $C_2$-profile of $c_1$ and the $B$-profile of $d_1 = f(c_1)$ are the same for every $c_1 \in C_1$.

It follows that for every $d_1 \in D_1$, 
$$\cost(\beta, d_1, D_2) = \cost(\beta,d_1,B)  \leq \cost(\beta, c_1, C_2) + \frac{\varepsilon^p}{t} \OPT^p.$$ (The term $\frac{\varepsilon^p}{t} \OPT^p$ is exactly the length of an interval in the definition of profiles).
Summing this upper bound on $\cost(\beta, d_1, B)$ over all $d_1 \in D_1$ and using~\eqref{eq:additivity},
we get
\begin{align*}
\cost_i(\beta,D_1,D_2) &= \sum_{d_1\in D_1} \cost_i(\beta,d_1,D_2)\\
&\leq
\sum_{d_1\in D_1} \Bigl(\cost_i(\beta,f^{-1}(d_1),C_2) + \frac{\varepsilon^p}{t} \OPT^p\Bigr) \\
&\leq 
\cost_i(\beta,C_1,C_2) + \varepsilon^p \OPT^p \leq 
\bigl((3c+2)^p - c^p + \varepsilon^p\bigr)\,\mathrm{OPT}^p.
\end{align*}
\end{proof}
\vspace*{-4mm}
\end{proof}







\begin{proof}[Proof of~\Cref{thm:FPT}]
Let $\beta$ be an optimal assignment for the center set $C$. Invoke~\Cref{lemma:findsets} to obtain the sets $D_1, D_2$. After obtaining the sets $D_1, D_2$, we close the centers $D_1$. Since $|D_1| = t$, it follows that we obtain a solution with at most $k$ centers.

We have $\|\cost(\beta)\|_{\infty} \le c^p\OPT^p$, and by~\Cref{lemma:findsets}, the additional cost incurred is at most $\|\cost(\beta, D_1, D_2) \|_{\infty} \leq ((3c + 2)^p - c^p + \varepsilon^p)\OPT^p$. Thus the $\ell_p$-objective of the resulting solution with $k$ centers is at most $\Bigl((3c + 2)^p - c^p + \varepsilon^p + c^p\Bigr)^{1/p} \OPT \leq (3c + 2 + \varepsilon) \OPT$ where in the final step we use the inequality $(a + b)^{1/p} \leq a^{1/p} + b^{1/p}$, which holds for every $p \geq 1,a,b \geq 0$.\end{proof}

By Theorem~\ref{thm:ghadiri-et-al} of~\cite{ghadiri2022constant}, we can find a $c = 5+2\sqrt{6}$ approximation solution with $k+\ell$ centers in polynomial time. By Lemma~\ref{lemma:findsets}, we can close $\ell$ centers in this bicriteria approximation solution to get a solution with $k$ centers in time $\bigl({\ell}(\frac{3c +2}{\varepsilon})^p\bigr)^{\ell^2} n^{O(1)}$. This proves~\Cref{thm:FPTmain}.

\section{Local search algorithm for \texorpdfstring{$(p,q)$}{(p,q)}-clustering}
\label{sec:local-search}
\subsection{Local search with scalar costs}
Following~\cite{Arya}
and~\cite{GT}, we first revisit the local search algorithm for $(k,p)$-clustering. We will later apply the results of this section to analyze the performance of the local search algorithm for $(p,q)$-Clustering.

Let $(X,d)$ be a metric space, and let $w(u)$ be non-negative vertex weights. We consider one step of the local search algorithm. Let $C = \{c_1,\dots,c_k\}$ be the current set of $k$ centers, and $C^* = \{c_1^*,\dots,c_k^*\}$ be a target set. Later, $C^*$ will be an optimal solution for the $(p, q)$-objective, but for now, we make no assumptions about it.

We define the cost of $C$ as 
\[
\cost(w) = \cost(w,C) =  \sum_{u\in X} w(u) d(u, C)^p,
\]
and the cost of $C^*$ as
\[
\cost^*(w) = \cost(w,C^*) = \sum_{u\in X} w(u) d(u, C^*)^p.
\]
Consider a swap operation $c_a \rightarrow c_b^*$ that replaces $c_a$ with $c_b^*$. We denote the resulting set of centers by
\[
C_{ab} = (C \setminus \{c_a\}) \cup \{c_b^*\},
\]
and its cost by $\cost_{ab}(w) = \cost(w,C_{ab}) = \sum_{u\in X} w(u) d(u, C_{ab})^p$.

Consider a probability distribution $\mathcal{D}$ on the set of swaps $c_a \rightarrow c_b^*$ (equivalently, on $[k] \times [k]$). The expected cost after a random swap drawn from $\mathcal{D}$ is $\mathbb{E}_{(a,b) \sim \mathcal{D}}[\cost_{ab}(w)]$.

We derive the following guarantees for local search. Conceptually, the proof closely follows the analyses in~\cite{Arya} and~\cite{GT}. However, the bounds we require are more intricate, as we need finer control over the distribution of swaps. We prove this lemma in \Cref{apx:local-search}.
\begin{lemma}\label{thm:basic-local-search}
Let $(X,d)$ be a metric space, and let $C$ and $C^*$ be two sets of centers. There exists a probability distribution $\mathcal{D}$ over swaps such that the following holds. For every $p\geq 1$, weights $w$, parameters $M \geq m > 0$, and coefficients $\xi_{ab}\in [m,M]$, we have.
\begin{enumerate}
\item For $\alpha = \frac{10Mp}{m}$, the following inequality holds.
\begin{equation}\label{eq:basic-local-search-1}
    \mathbb{E}_{\mathcal{D}}[\xi_{ab}\cdot(\cost_{ab}(w) - \cost(w))] \leq \frac{m}{2k}(\alpha^p \cost^*(w) - \cost(w))
\end{equation}    
\item We also have,
\begin{equation}\label{eq:basic-local-search-2}
\mathbb{E}_{\mathcal{D}}[\xi_{ab}\cdot( \cost_{ab}(w) - \cost(w))] \leq \frac{(1+2^{2p})M}{k} (\cost^*(w)+\cost(w))
\end{equation}
    \item Finally, for every swap $c_a \rightarrow c_b^*$ in the support of $\mathcal{D}$ and every $\delta_0\in (0,1/4)$, we have
\begin{equation}\label{eq:basic-local-search-3}
\cost_{ab}(w)  \leq (1+\theta_0)\cost^*(w) + (1+ \delta_0) \cost(w),
\end{equation}
where $\theta_0 = 2^p(5p/(4\delta_0))^{p-1}$.
\end{enumerate}
\end{lemma}

On a high level, the necessity for such guarantees is as follows. In $k$-median, the objective is linear, and the effects of the change in cost due to a single swap can be directly accounted for. However, in $(p,q)$-Clustering, the objective involves two layers of non-linearity (the inner and outer norms). To deal with this non-linearity, we can write the change in the objective for each swap $c_a \rightarrow c_b^*$ as a linear change with multipliers $\xi_{ab}$ using the mean value theorem and convexity. Here the multipliers $\xi_{ab}$ are bounded in the range $[m,M]$ for some appropriate choice of $m$ and $M$. The above lemma then assumes this boundedness of $\xi_{ab}$ and derives guarantees on the change in $\cost(w)$, which in turn bounds the change in the objective.

\subsection{Single swap analysis w.r.t.\ the \texorpdfstring{$(p,q)$}{(p,q)}-objective}
Now we consider the local search algorithm for $(p,q)$-Clustering --
the goal is to minimize the following objective function:
\[
\left(\sum_{i=1}^{\ncoord} \left(\sum_{u\in X} w_i(u) d(u, C)^p \right)^{q/p}\right)^{1/q} = \left(\sum_{i=1}^{\ncoord} \cost(w_i,C)^{q/p}\right)^{1/q}
\]
From now on, let $C^*$ be an optimal solution for the $(p,q)$-Clustering problem. We assume that swap $c_a \rightarrow c_b^*$ is sampled from distribution $\cal D$ constructed in Theorem~\ref{thm:basic-local-search}.
Define $r= q/p$ and $\varphi(t) = t^{r}$, and
\begin{align*}
    z_i &= \cost(w_i), &
    z'_i &= \cost_{ab}(w_i), &
    z^*_i &= \cost^*(w_i) \\
    \Phi &= \sum_{i=1}^{\ncoord} \varphi(z_i), &
    \Phi_{ab} &= \sum_{i=1}^{\ncoord} \varphi(z'_i), &
    \Phi^* &= \sum_{i=1}^{\ncoord} \varphi(z^*_i)
\end{align*}
Note that  $(\Phi^*)^{1/q}$ is the cost of the optimal solution, $\Phi^{1/q}$ is the cost of the current solution, and $\Phi_{ab}^{1/q}$ is the cost of the solution after a swap $c_a \to c_b^*$. Also, let 
\[\Delta_i = \E{\varphi(z_i') - \varphi(z_i)}.\]

Fix a coordinate $i$. Now we get an upper bound for $\varphi(z'_i)$ by approximating $\varphi$ with a piecewise linear function. Let $\hat{z}_i$ be the maximum of $z_i'$ over all swaps in the support of $\cal D$ (we will obtain an upper bound on $\hat z_i$ later using Theorem~\ref{thm:basic-local-search}, item 3). 
Let $M = \frac{\varphi(\hat{z}_i) - \varphi(z_i)}{\hat{z}_i - z_i} \leq \varphi'(\hat z_i)$ (by the Mean Value Theorem and monotonicity of $\varphi'$) and $m = \varphi(z_i)/z_i = z_i^{r-1}$.
Since $\varphi$ is convex, we have (see Figure~\ref{fig:approximation-for-phi})

\begin{figure}[!htbp]
\centering
\begin{tikzpicture}[scale=1.2,font=\footnotesize]
    \draw[->] (0,0) -- (2.1,0) node[right] {$z'_i$};
    \draw[->] (0,0) -- (0,2.1) node[right] {$\varphi(z'_i)$};

    \draw[red, domain=0:1.9, samples=100] plot (\x, {0.5*\x*\x});

    \draw[blue] (0,0) -- (1,0.5);
    \draw[blue] (1,0.5) -- (1.7,1.445);

    \draw[dotted] (1,0) -- (1,0.5);
    \draw[dotted] (1.7,0) -- (1.7,1.445);
    
    \draw[dotted] (0,0.5) -- (1,0.5);
    \draw[dotted] (0,1.445) -- (1.7,1.445);

    \filldraw[black] (1,0.5) circle (0.02);

    \node[below] at (0,0) {\strut $0$};   

    \draw (1,0) -- (1,-0.05);
    \node[below] at (1, 0) {\strut $z_i$};

    \draw (1.7,0) -- (1.7,-0.05);
    \node[below] at (1.7,0.0) {\strut $\hat{z}_i$};

    \draw (0,0.5) -- (-0.05,0.5);
    \node[left] at (0,0.5) {$\varphi(z_i)$};
    \draw (0,1.445) -- (-0.05,1.445);
    \node[left] at (0,1.445) {$\varphi(\hat{z}_i)$};    
\end{tikzpicture}
\caption{We approximate $\varphi(z'_i)$ on the interval $[0, \hat{z}_i]$ using a piecewise linear function consisting of two segments: the first with slope $m$, and the second with slope $M$.}
\label{fig:approximation-for-phi}
\end{figure}
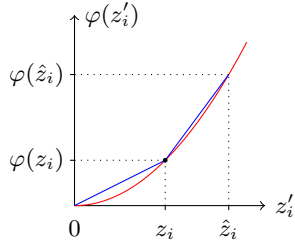

\[
\varphi(z_i') \leq 
\begin{cases}
    \varphi(z_i) + m \cdot (z_i' - z_i), & \text{if } z'_i < z_i \\
    \varphi(z_i) + M\cdot (z'_i - z_i), & \text{if } z'_i \geq z_i 
\end{cases}
\]
Define coefficients $\xi_{ab}$ as follows
\[
\xi_{ab} = 
\begin{cases}
m, & \text{if } z_i' < z_i\\
M, & \text{otherwise}
\end{cases}
\]
Then $\varphi(z_i') \leq \varphi(z_i) + \xi_{ab}\cdot (z_i' - z_i)$, and we have
$\Delta_i\equiv\E{\varphi(z'_i) - \varphi(z_i)} \leq \E{\xi_{ab} \cdot (z_i' - z_i)}$.

\paragraph*{Now we are in a position to use~\Cref{thm:basic-local-search}, using which we derive the following key inequality, whose proof is deferred to~\Cref{sec:prooflocalsearch}.}

\begin{equation}\label{eq:localsearchguarantee}
\Delta_i\leq \frac{1}{4k}\left(\beta^q \varphi(z^*_i) - \varphi(z_i)\right)
\end{equation}
with $\beta = O(q)$.
Adding up the inequality over all coordinates $i\in [\ncoord]$, we get
\[\E{\Phi_{ab} - \Phi} = \sum_{i=1}^{\ncoord} \Delta_i \leq \frac{1}{4k}\left(\beta^q \Phi^* - \Phi\right).
\]
In particular, this means that there is always a swap $c_a\rightarrow c_b^*$, for which 
\[\Phi_{ab} - \Phi \leq \frac{1}{4k}\left(\beta^q \Phi^* - \Phi\right).\]

\subsection{Putting everything together}
Consider the local search algorithm that at each iteration finds a swap $a\to b$ with $a\in C$ and $b\in F \setminus C$ that minimizes the potential function $\Phi$ the most, where $C$ is the current set of open centers and $F$ is the candidate facility set. The algorithm terminates when no swap decreases the current potential $\Phi$ by more than $\frac{\Phi}{8k}$.

\subsubsection{Bounding the number of iterations}  
Suppose that $\Phi > 2\beta^q \Phi^*$. Then, for an optimal swap $(a, b)$, we have  
\[
\Phi_{ab} - \Phi \leq -\frac{\Phi}{8k}
\]
and  
\[
\Phi_{ab} - \Phi^* \leq (\Phi - \Phi^*) - \frac{\Phi}{8k} \leq \left(1 - \frac{1}{8k}\right)(\Phi - \Phi^*).
\]

Thus when the algorithm stops, we must have $\Phi \leq 2\beta^q \Phi^*$. Till the algorithm stops, in every iteration, the above inequality holds. It follows that after at most $8k \log_e(\Phi / \Phi^*)$ iterations, the algorithm will find a clustering of cost at most $2^{1/q} \cdot\beta\cdot (\Phi^*)^{1/q}$.

As is standard, by truncating distances and rounding weights, we may assume that all distances and weights are integers bounded polynomially in $n$. Therefore, the cost of any clustering is polynomial in $n$; hence, $\log_e(\Phi / \Phi^*)$ is bounded by $O(q \log n)$. We conclude that the algorithm returns a solution of cost at most $O(q)$ in at most $O(qk \log n)$ iterations.

\subsubsection{Running time per iteration} In each iteration, we maintain the current set of centers $C$. For each point $u$ we maintain its closest and second-closest centers $s_1(u)$ and $s_2(u)$ in $C$. At the beginning of every iteration, these quantities can be computed in time $O(nk)$. We can compute the change in cost due to a swap $a \rightarrow b$ for each point $u$ using $s_1(u), s_2(u)$ and the distances $d(u,a), d(u,b)$. Then we add these costs for each group in time $O_{p,q}(n\ell)$ and re-compute the potential. Thus, for every swap $a \rightarrow b$ we spend time $O_{p,q}(n\ell)$ and there are $O(nk)$ swaps, giving a per-iteration running time of $O(n^2k\ell)$. Since there are $O(kq\log n)$ iterations, the running time is at most $O_{p,q}(n^2 \ell k^2 \log n)$. This concludes the proof of~\Cref{thm:local_search}.

\subsection{Proof of Lemma~\ref{thm:basic-local-search}}
\label{apx:local-search}
In this Section, we prove Lemma~\ref{thm:basic-local-search}.
\begin{proof}
\textbf{Defining the distribution of swaps.} For each center $c_b^* \in C^*$, let $\gamma(c_b^*)$ denote the nearest center in $C$ (see Figure~\ref{fig:exchange}). We partition the centers in $C$ based on the number of preimages under $\gamma$:
\begin{figure}[!htbp]
\centering
\begin{tikzpicture}[scale=0.8, every node/.style={font=\small}]
\usetikzlibrary{arrows.meta, positioning}
\usetikzlibrary{decorations.pathreplacing}

\def\squareSize{0.6}
\def\circleRadius{0.36}
\def\n{6}
\def\spacing{2*\squareSize}
\def\xleft{0}
\def\xright{3.5}
\def\ymid{\spacing*(\n-1)/2 + 2}

\node at ({(\xleft+\xright)/2}, {\ymid + 1.5}) {$\gamma$};

\draw[decorate, decoration={brace, mirror, amplitude=5pt}, thick] 
  ({\xleft - \squareSize/2 - 0.3}, {\spacing*(\n-0.7)}) -- 
  ({\xleft - \squareSize/2 - 0.3}, {\spacing*(\n-2.3)}) node[midway, left=6pt] {$B_1$};

\draw[decorate, decoration={brace, mirror, amplitude=5pt}, thick]
  ({\xleft - \squareSize/2 - 0.3}, {\spacing*(\n-2.7)}) -- 
  ({\xleft - \squareSize/2 - 0.3}, -0.3) node[midway, left=6pt] {$B_{\star}$};

\draw[decorate, decoration={brace, amplitude=5pt}, thick] 
  ({\xright + \circleRadius + 0.3}, {\spacing*(\n -0.7)}) -- 
  ({\xright + \circleRadius + 0.3}, {\spacing*(\n - 2.3)}) node[midway, right=6pt] {$A_1$};

\draw[decorate, decoration={brace, amplitude=5pt}, thick]
  ({\xright + \circleRadius + 0.3}, {\spacing*(\n - 2.7)}) -- 
  ({\xright + \circleRadius + 0.3}, {\spacing*(\n - 4.3)}) node[midway, right=6pt] {$A_{\star}$};

\draw[decorate, decoration={brace, amplitude=5pt}, thick]
  ({\xright + \circleRadius + 0.3}, {\spacing*(\n - 4.7)}) -- 
  ({\xright + \circleRadius + 0.3}, {\spacing*(\n - 6.3)}) node[midway, right=6pt] {$A_0$};

\foreach \i in {1,...,6} {
  \pgfmathsetmacro{\y}{\spacing*(6-\i)}
  
  \draw[thick] ({\xleft - \squareSize/2}, \y - \squareSize/2) rectangle ({\xleft + \squareSize/2}, \y + \squareSize/2);
  \ifnum\i=1
    \node at ({\xleft  }, \y) [anchor=center] {$c_1^*$};
  \else\ifnum\i=6
    \node at ({\xleft }, \y ) [anchor=center] {$c_k^*$};
  \fi\fi

  \draw[thick] (\xright,\y) circle (\circleRadius);
  \ifnum\i=1
    \node at ({\xright }, \y) [anchor=center] {$c_1$};
  \else\ifnum\i=6
    \node at ({\xright  }, \y) [anchor=center] {$c_k$};
  \fi\fi
}

\foreach \s/\t in {1/1, 2/2, 3/3, 4/3, 5/4, 6/4} {
  \pgfmathsetmacro{\yfrom}{\spacing*(6-\s)}
  \pgfmathsetmacro{\yto}{\spacing*(6-\t)}
  \draw[-{Latex[length=2.2mm]}, thin] 
    ({\xleft + \squareSize/2}, \yfrom) -- ({\xright - \circleRadius}, \yto);
}
\end{tikzpicture} 
\qquad\qquad
\begin{tikzpicture}[scale=0.8, every node/.style={font=\small}]
\usetikzlibrary{arrows.meta, positioning}
\usetikzlibrary{decorations.pathreplacing}

\def\squareSize{0.6}
\def\circleRadius{0.36}
\def\n{6}
\def\spacing{2*\squareSize}
\def\xleft{0}
\def\xright{3.5}
\def\ymid{\spacing*(\n-1)/2 + 2}

\draw[decorate, decoration={brace, mirror, amplitude=5pt}, thick] 
  ({\xleft - \squareSize/2 - 0.3}, {\spacing*(\n-0.7)}) -- 
  ({\xleft - \squareSize/2 - 0.3}, {\spacing*(\n-2.3)}) node[midway, left=6pt] {$B_1$};

\draw[decorate, decoration={brace, mirror, amplitude=5pt}, thick]
  ({\xleft - \squareSize/2 - 0.3}, {\spacing*(\n-2.7)}) -- 
  ({\xleft - \squareSize/2 - 0.3}, -0.3) node[midway, left=6pt] {$B_{\star}$};

\draw[decorate, decoration={brace, amplitude=5pt}, thick] 
  ({\xright + \circleRadius + 0.3}, {\spacing*(\n -0.7)}) -- 
  ({\xright + \circleRadius + 0.3}, {\spacing*(\n - 2.3)}) node[midway, right=6pt] {$A_1$};

\draw[decorate, decoration={brace, amplitude=5pt}, thick]
  ({\xright + \circleRadius + 0.3}, {\spacing*(\n - 2.7)}) -- 
  ({\xright + \circleRadius + 0.3}, {\spacing*(\n - 4.3)}) node[midway, right=6pt] {$A_{\star}$};

\draw[decorate, decoration={brace, amplitude=5pt}, thick]
  ({\xright + \circleRadius + 0.3}, {\spacing*(\n - 4.7)}) -- 
  ({\xright + \circleRadius + 0.3}, {\spacing*(\n - 6.3)}) node[midway, right=6pt] {$A_0$};

\foreach \i in {1,...,6} {
  \pgfmathsetmacro{\y}{\spacing*(6-\i)}
  
  \draw[thick] ({\xleft - \squareSize/2}, \y - \squareSize/2) rectangle ({\xleft + \squareSize/2}, \y + \squareSize/2);
  \ifnum\i=2
    \node at ({\xleft  }, \y) [anchor=center] {\textcolor{red}{$c_1^*$}};
  \else\ifnum\i=3
    \node at ({\xleft }, \y ) [anchor=center] {\textcolor{blue}{$c_b^*$}};
  \fi\fi

  \draw[thick] (\xright,\y) circle (\circleRadius);
  \ifnum\i=2
    \node at ({\xright }, \y) [anchor=center] {\textcolor{red}{$c_a$}};
  \else\ifnum\i=5
    \node at ({\xright  }, \y) [anchor=center] {\textcolor{blue}{$c_a$}};
  \fi\fi
}

\draw[blue, thick, ->] ({\xright - \circleRadius}, {\spacing*(6 - 5)}) -- ({\xleft + \squareSize/2}, {\spacing*(6 - 3)});
\draw[red, thick, ->] ({\xright - \circleRadius}, {\spacing*(6 - 2)}) -- ({\xleft + \squareSize/2}, {\spacing*(6 - 2)});

\end{tikzpicture} 

\caption{The left figure shows map $\gamma$ that maps every optimal center $c_i$ to the closest current center $c_j$. This map defines the partition of the optimal centers to sets $B_1$ and $B_{\star}$ and the current centers to sets $A_0$, $A_1$, and $A_{\star}$. The right figure illustrates two possible types of swaps $c_a \rightarrow c_b$.}
\label{fig:exchange}

\end{figure}
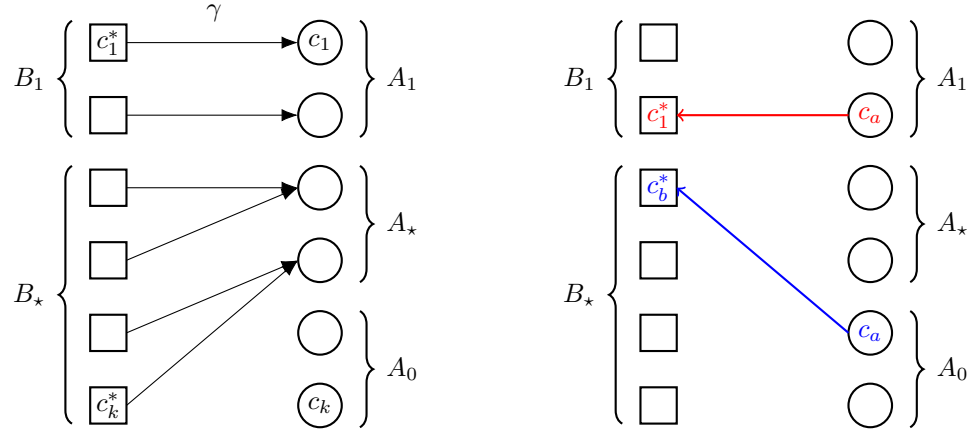
\begin{itemize}
\item $A_1$: centers with exactly one preimage,
\item $A_0$: centers with no preimages,
\item $A_{\star}$: centers with two or more preimages.
\end{itemize}
Define $B_1 = \gamma^{-1}(A_1)$ and $B_{\star} = \gamma^{-1}(A_{\star})$. We observe:
\begin{itemize}
    \item $\gamma$ is a one-to-one map between $B_1$ and $A_1$,
    \item $\gamma$ maps $B_{\star}$ onto $A_{\star}$, with every $c \in A_{\star}$ having at least two preimages.
\end{itemize}
Denote $r = |A_1|$. From the observations above, we get $|B_1| = |A_1| = r$ and $|A_{\star}| \leq |B_{\star}|/2 = (k - r)/2$. Thus, $|A_0| = k - |A_1| - |A_\star| \geq (k-r)/2$.
Now we are ready to define a probability distribution $\mathcal{D}$ over swaps.
\begin{enumerate}
    \item With probability $\frac{r}{k}$, select $c_b^* \in B_1$ uniformly at random and set $c_a = \gamma(c_b^*)\in A_1$,
    \item With probability $\frac{k-r}{k}$, select $c_a$ uniformly from $A_0$ and $c_b^*$ uniformly from $B_{\star}$.
\end{enumerate}
This yields a random swap $c_a \rightarrow c_b^*$. 
Note that $b$ is uniformly distributed in $[k]$ and therefore the probability that $b = b_0$ is $1/k$  for every fixed $b_0\in [k]$. In general, $a$ is not uniformly distributed; however, the probability that $a = a_0$ for every fixed $a_0\in [k]$ is at most $\max(\frac{r}{k}\cdot\frac{1}{|A_1|}, \frac{k-r}{k}\cdot\frac{1}{|A_0|}) \leq \max(1/k, 2/k) = 2/k$:
\begin{equation}\label{eq:prob}
    \Prob{a = a_0} \leq 2/k \qquad\text{and} \qquad \Prob{b = b_0} = 1/k 
\end{equation}

We will need below that for all $t \neq b$,
\begin{equation}\label{eq:no-hit}
\gamma(c_t^*) \neq c_a.
\end{equation}
Indeed, in item 1, $\gamma^{-1}(c_a) = \{c_b^*\}$; in item 2, $\gamma^{-1}(c_a) = \varnothing$. In either case, $c_t^*\notin \gamma^{-1}(c_a)$.

\subsubsection{Upper bounding connection cost}
Let $x_u = d(u,C)$ denote the distance from $u$ to the closest center in the current solution $C$, $x_u' = d(u,C_{ab})$ the updated distance after the swap, and $x_u^* = d(u,C^*)$ the distance to the closest center in $C^*$. Let $V_a$ be the Voronoi cell of $c_a$ in $C$, and $V_b^*$ the Voronoi cell of $c_b^*$ in $C^*$. For $u \in V_b^*$,
    \[
    x_u^{\prime p} - x_u^p \leq x_u^{*p} - x_u^p.
    \]
For $u \in V_a \setminus V_b^*$, assume $u \in V_t^*$ for some $t \neq b$. By~\eqref{eq:no-hit}, $\gamma(c_t^*) \neq c_a$; thus, $\gamma(c_t^*)\in C_{ab}$ and
\begin{figure}
\centering
\begin{tikzpicture}[scale=1.5, every node/.style={font=\small}]
  \coordinate (p1) at (-0.5, 1.2);      
  \coordinate (p2) at (1, 2);      
  \coordinate (p3) at (-1, 0);     
  \coordinate (p4) at (1, 1);      

  \draw[thick] (p1) -- (p4); 
  \draw[thick, red] (p1) -- (p3); 
  \draw[thick, red] (p3) -- (p4); 
  \draw[dotted, thick] (p3) -- (p2);

  \filldraw[black] (p1) circle (0.03);
  \node[above left] at (p1) {$u$};
    
  \filldraw[black] (p2) circle (0.03);
  \node[right] at (p2) {$c_a$};
    
  \filldraw[black] (p3) circle (0.03);
  \node[left] at (p3) {$c_t^*$};
    
  \filldraw[black] (p4) circle (0.03);
  \node[below right] at (p4) {$\gamma(c_t^*) \in C_{ab}$};    
\end{tikzpicture}
\qquad
\begin{tikzpicture}[scale=1.5, every node/.style={font=\small}]
  \coordinate (p1) at (-0.5, 1.2);      
  \coordinate (p2) at (1, 2);      
  \coordinate (p3) at (-1, 0);     
  \coordinate (p4) at (1, 1);      

  \draw[thick, brown] (p1) -- (p3); 
  \draw[thick, brown] (p1) -- (p2); 
  \draw[dotted, thick] (p3) -- (p2);

  \filldraw[black] (p1) circle (0.03);
  \node[above left] at (p1) {$u$};
    
  \filldraw[black] (p2) circle (0.03);
  \node[right] at (p2) {$c_a$};
    
  \filldraw[black] (p3) circle (0.03);
  \node[left] at (p3) {$c_t^*$};
    
  \filldraw[black] (p4) circle (0.03);
  \node[below right] at (p4) {$\gamma(c_t^*) \in C_{ab}$};    
\end{tikzpicture}
\caption{We assign point $u$ to $\gamma(c_t^*)$, which is guaranteed to be in $C_{ab}$. The distance from $u$ to $\gamma(c_t^*)$ is upper bounded by $d(u,c_t^*) + d(c_t^*, \gamma(c_t^*))$. Further, $d(c_t^*, \gamma(c_t^*)) \leq d(c_t^*, c_a)$, since $\gamma(c_t^*)$ is the closest current center to $c_t^*$. In turn, $d(c_t^*, c_a)$ is upper bounded by $d(c_t^*, u) + d(u,c_a)$.}
\label{fig:charging}
\end{figure}
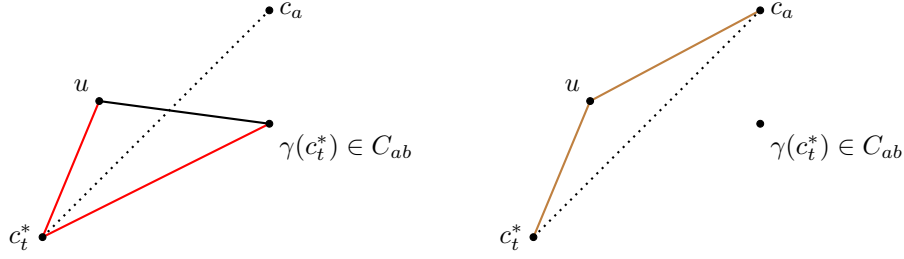
\begin{align*}
    x_u' &= d(u, C_{ab}) 
    \leq d(u, \gamma(c_t^*)) 
    \leq d(u, c_t^*) + d(c_t^*, \gamma(c_t^*)) 
    \leq d(u, c_t^*) + d(c_t^*, c_a) \\
    &\leq d(u, c_t^*) + \big(d(c_t^*, u) + d(u, c_a)\big) 
    = 2x_u^* + x_u.
\end{align*}
Here, we used that $d(c_t^*, \gamma(c_t^*)) \leq d(c_t^*, c_a)$, 
since $\gamma(c_t^*)$ is the nearest center in $C$ to $c_t^*$.

Applying Claim~\ref{claim:MMR} with $\delta = (1+\varepsilon)^{p-1} - 1$ and $\theta = 2^p (1+1/\varepsilon)^{p-1}$, we get
\[x^{\prime p}_u \leq (x_u + 2x_u^*)^p \leq (1+\delta) x_u^p + \theta x_u^{*p}.\]
Thus,
\[x^{\prime p}_u - x_u^p\leq  \delta x_u^p + \theta x_u^{*p}.\]

Finally, if $u\notin V_a \cup V^*_b$, then $x_u' \leq x_u$.
We have,
\begin{align}
\cost_{ab}(w) - \cost(w) &\leq \sum_{u\in V_b^*} w_u \left( x_u^{\prime p} - x_u^p \right) + \sum_{u\in V_a\setminus V_b^*} w_u \left( x_u^{\prime p} - x_u^p \right) \\
&\leq 
\sum_{u\in V_b^*} w_u \left( x_u^{*p} - x_u^p \right) + \sum_{u\in V_a} w_u \left( \delta x_u^p + \theta x_u^{*p} \right). \label{eq:cost-always}
\end{align}
We now derive an upper bound for $\xi_{ab}\cdot( \cost_{ab}(w) - \cost(w) )$. Since $0< m \leq \xi_{ab} \leq M$, we can replace $\xi_{ab}$ with $M$ in positive terms and with $m$ in negative terms. We have,
\[\xi_{ab}\cdot( \cost_{ab}(w) - \cost(w) ) \leq 
\sum_{u\in V_b^*} w_u \left( M x_u^{*p} - m x_u^p \right) + M\sum_{u\in V_a} w_u \left( \delta x_u^p + \theta x_u^{*p} \right)
\]
Now we assume that the swap $c_a\to c_b^*$ is sampled from the distribution $\cal D$ defined above. We get
\begin{align*}
&\E{\xi_{ab}\cdot( \cost_{ab}(w) - \cost(w) )} \\
&\qquad \leq\sum_{u\in V} \Prob{u\in V_b^*} \cdot w_u \left( M x_u^{*p} - m x_u^p \right) + M\sum_{u\in V}\Prob{u\in V_a} \cdot w_u \left( \delta x_u^p + \theta x_u^{*p} \right)
\end{align*}
From~\ref{eq:prob}, we get $\Prob{u\in V_b^*} = \frac{1}{k}$ and $\Prob{u\in V_a} \leq \frac{2}{k}$ and thus
\begin{equation}\label{eq:main-expected-change}
k \E{\xi_{ab}\cdot( \cost_{ab}(w) - \cost(w) )} \leq (1 + 2\theta) M \cost^*(w) - \left(m - 2\delta M \right) \cost(w).
\end{equation}
We will use (\ref{eq:main-expected-change}) to derive (\ref{eq:basic-local-search-1}), (\ref{eq:basic-local-search-2}), and (\ref{eq:basic-local-search-3}). We will use different values of parameters $\varepsilon$, $\delta$, and $\theta$ in these derivations. We will need the following lemma.

\begin{lemma}\label{lem:delta-theta}
For every $p \geq 1$ and every value of $\delta_0\in (0,1/4)$, there exists $\varepsilon > 0$ such that $\delta = (1+\varepsilon)^{p-1} - 1 \leq \delta_0$ and $\theta \leq 2^p\left(\frac{5p}{4\delta_0}\right)^{p-1}$.
\end{lemma}
\begin{proof}
If $p = 1$, we let $\varepsilon = 1$, resulting in $\delta = 0$ and $\theta = 2$, and we are done. If $p > 1$, we choose $\varepsilon$ so that $\delta = \delta_0$. Applying the Mean Value Theorem, we get 
\[\varepsilon = (1 + \delta_0)^{1/(p-1)} - 1 \geq \delta_0 \cdot \min_{t\in(0,\delta_0)}\frac{(1 + t)^{1/(p-1)-1}}{p-1} \geq \frac{4}{5}\cdot \frac{\delta_0}{p-1},\]
where we used that $1 + t\leq 1 + \delta_0 \leq 5/4$ and $1/(p-1) - 1 \geq -1$.
Therefore,
\[\theta= 2^p (1+1/\varepsilon)^{p-1} \leq 2^p \left(\frac{5(p-1) + 4\delta_0}{4\delta_0}\right)^{p-1} \leq 2^p \left(\frac{5p}{4\delta_0}\right)^{p-1}.\]
\end{proof}

First we establish upper bound (\ref{eq:basic-local-search-1}).
We choose parameters according to Lemma~\ref{lem:delta-theta} with $\delta_0 = m/(4M)$. Then $\theta \leq 2^p (5pM/m)^{p-1}$.
Since $0.5 \cdot 2^p (5Mp/m)^{p-1} \geq 2^{p-1} \geq 1$, we have
\[(1+2\theta) M \leq 2.5 \cdot 2^p \left(\frac{5Mp}{m}\right)^{p-1} \cdot \frac{M}{m} \cdot m \leq \frac{1}{2}\left(\frac{10Mp}{m}\right)^p m.\]
We conclude that
\[\E{\xi_{ab}\cdot( \cost_{ab}(w) - \cost(w) )} \leq \frac{m}{2k}\left(\left(\frac{10Mp}{m}\right)^p \cost^*(w) - \cost(w)\right).\]

Now we prove (\ref{eq:basic-local-search-2}). We set $\varepsilon = 1$; then $\theta = 2^{2p-1}$ and $\delta = 2^{p-1} -1$. From~(\ref{eq:main-expected-change}), we get
\begin{align*}
&k\E{\xi_{ab}\cdot( \cost_{ab}(w) - \cost(w) )} \\
&\leq (1+2\theta) M \cost^*(w) + 2\delta M \cost(w) \leq (1+2^{2p})M (\cost(w)+\cost^*(w)))
\end{align*}

Finally, for every swap $c_a \to c_b^*$ in the support of $\cal D$, we get the following upper bound from~\eqref{eq:cost-always}, 
\begin{align*}\cost_{ab}(w) - \cost(w) &\leq \sum_{u\in V_b^*} w_u \left( x_u^{*p} - x_u^p \right) + \sum_{u\in V_a} w_u \left( \delta x_u^p + \theta x_u^{*p} \right)\\
&\leq 
\sum_{u\in V} w_u x_u^{*p} + \sum_{u\in V} w_u  \left( \delta x_u^p + \theta x_u^{*p} \right) \leq (1 + \theta) \cost^*(w) + \delta \cost(w) .
\end{align*}
Choosing $\varepsilon$ according to Lemma~\ref{lem:delta-theta}, we get the desired bound.
\end{proof}

\subsection{Proof of Inequality~\ref{eq:localsearchguarantee}}\label{sec:prooflocalsearch}
Applying~\Cref{thm:basic-local-search}, we get for $\alpha = 10Mp/m$,
\begin{align}
\Delta_i  &\leq \frac{m}{2k}\left(\alpha^p z^*_i - z_i\right)\label{eq:Delta-large}\\
\Delta_i  &\leq \frac{(1+2^{2p})M}{k}(z_i+z^*_i).\label{eq:Delta-small}
\end{align}
From~\Cref{thm:basic-local-search}, item 3, with $\delta_0 = 1/(5r)$ and \[
\theta_0 = 2^p\left(\frac{25pr}{4}\right)^{p-1},
\]
we get
\begin{equation}\label{eq:hat-zi}
\widehat z_i \le (1+\theta_0)z_i^*+(1+\delta_0)z_i
\le
\left(1+\frac1{5r}\right)z_i+(1+\theta_0)z_i^* .
\end{equation}
Fix a threshold \(
T=\max\left(\frac{5r(1+\theta_0)}4,\;2(10eq)^p\right)=O(q)^p
\).
Consider two cases.
\paragraph*{Current cost $z_i$ is small: $z_i \leq T z_i^*$.}
From~\eqref{eq:hat-zi} and $1+\theta_0 \leq \frac{4}{5r}T$,
\[\hat{z}_i \leq \frac65 T z_i^* + \frac{4}{5} Tz_i^* = 2Tz_i^*.\]
Accordingly, $M \leq \varphi'(\hat z_i) = r\hat{z}_i^{r-1} \leq r (2Tz_i^*)^{r-1}$. Also, $m \leq M$.
By \eqref{eq:Delta-small}
\begin{align*}
\Delta_i &\leq \frac{(4^{p}+1)M}{k}(z_i^*+z_i) =
\frac{(4^{p}+1)}{k}\left(M\cdot (z_i^*+ z_i +\frac{m}{M}z_i) - m\cdot z_i\right) \\
&\leq 
\frac{(4^{p}+1)}{k}((2T+1)Mz_i^*-mz_i) \leq \frac{1}{2k}\left(4^{p+2}\cdot r \cdot  (2T)^r\varphi(z_i^*) - \varphi(z_i)\right).
\end{align*}
Letting $\beta = (2\cdot 4^{p+2}\cdot r \cdot (2T)^r)^{1/q} = O(q)$, we get 
\[\Delta_i\leq \frac{1}{4k}\left(\beta^q \varphi(z^*_i) - \varphi(z_i)\right).\]

\paragraph*{Current cost $z_i$ is large: $z_i \geq Tz_i^*$.}
We show that $2\alpha^p \leq T$.
From~\eqref{eq:hat-zi} and $1+\theta_0 \leq \frac{4}{5r}T$, we have
\[\hat z_i \leq \left(1+ \frac{1}{5r} + \frac{1+\theta_0}{T}\right) z_i\leq \left(1+\frac1r\right) z_i.\]
Hence, $M\leq \varphi'(\hat z_i) = r {\hat z}_{i}^{r-1}\leq er z_i^{r-1}$ and $\alpha \leq 10 p\cdot (M/m) \leq 10 p \cdot (er) = 10 eq$. 
Using that $z_i \geq Tz_i^* \geq 2 (10eq)^p z_i^* \geq 2\alpha^p z_i^*$ and $\varphi(z_i) = m z_i$, we get from \eqref{eq:Delta-large},
\[\Delta_i \leq \frac{m}{2k}(\alpha^p z_i^* - z_i) \leq -\frac{m}{2k}\cdot \frac{z_i}{2} = -\frac{1}{4k}\varphi(z_i).\]

\section{Socially fair facility location}
\label{sec:facility}
In this section, we consider two variants of the Socially Fair Facility Location problem and provide approximation algorithms for them.

We are given a set of $n$ points $X$ and a metric $d$ on it, forming a metric space $(X, d)$, as well as a set of possible facility locations $F$. There are $\ell$ groups on $X$ represented by a weight function $w: X \to \mathbb{R}_{\geq 0}^\ell$ where $w_i(u) \geq 0$  denotes the weight of client $u$ in group $i$. Each facility $c \in F$ is associated with a facility opening cost $f(c)$. We consider the following two settings with different types of facility opening costs: (1) group homogeneous facility cost; (2) group heterogeneous facility cost.

\paragraph*{Group homogeneous facility costs.}
Each facility $c\in F$ has a scalar opening cost $f(c)\ge 0$, paid equally across groups. The objective is to choose $C\subseteq F$ minimizing
\[
\mathrm{FL}_{\mathrm{hom}}(w,f,C):= \;:=\;
\left(\sum_{i=1}^{\ell} \cost(w_i,C)^{q/p}\right)^{1/q} + \sum_{c \in C} f(c),
\]
where $\cost(w_i,C) = \sum_{u\in X} w_i(u)\, d(u,C)^p$.

\cite{abbasi2021fair} studied this objective with $p=1$ and $q=\infty$ and gave a $4$-approximation algorithm for it.  
\paragraph*{Group heterogeneous facility costs.}
Each facility $c\in F$ has a vector opening cost $f(c)\in\mathbb{R}_{\ge 0}^{\ell}$, where group $i$ pays $f_i(c)$ to open $c$.
\footnote{We interpret $f_i(c)$ as a group-specific \emph{fixed} charge for opening facility $c$ (e.g., taxes/subsidies or administrative overhead). Thus, group $i$ pays $f_i(c)$ whenever $c$ is opened, regardless of whether any group-$i$ client is assigned to $c$; this keeps opening costs dependent only on the opened set $C$, as in standard fixed-charge facility location.}

In this case, the objective is\footnote{If the heterogeneous opening costs are \emph{group-uniform}, i.e., $f_i(c)=f(c)$ for all $i\in[\ell]$ and $c\in F$, then the objective becomes
$\left(\sum_{i=1}^{\ell}\big(\cost(C,w_i)^{1/p}+\sum_{c\in C} f(c)\big)^q\right)^{1/q}$.
We refer to this as the \emph{group-uniform cost} variant.
This objective is generally \emph{different} from the group homogeneous facility costs above.
The two coincide when $q=\infty$, but they differ for finite $q$.
By using our heterogeneous convex relaxation with $f_i(c)=f(c)$ and applying the rounding scheme used for the homogeneous costs, we obtain an $(8+\varepsilon)$-approximation for this group-uniform cost variant.}

\[
\mathrm{FL}_{\mathrm{het}}(w, f, C)
\;:=\;
\left(\sum_{i=1}^{\ell}
\left(\cost(w_i,C)^{1/p} + \sum_{c\in C} f_i(c)\right)^q
\right)^{1/q}.
\]

We get the following results for homogeneous and heterogeneous facility costs.

\FLscalar*
\FLvector*

\begin{proof}[Proof of Theorem~\ref{thm:fl_scalar}]
    We provide a convex relaxation and rounding algorithm as follows. 
    We use the following convex programming relaxations inspired by~\cite{chlamtavc2022approximating}. We use the convex program below. 
    \begin{align}
    \text{min} \quad & 
    H + \sum_{c \in F} y_c \cdot f(c) \nonumber \\
    \text{s.t.} \quad 
    & \sum_{i=1}^{\ell} z_i \leq H^q \nonumber \\
    & z_i \geq \left( \sum_{u \in X} w_j(u) \sum_{c \in F} d(u, c)^p \cdot x_{u,c} \right)^{q/p} 
    &\quad& \forall i \in [\ell] \label{eq:fl_constraint} \\
    & x_{u,c} \leq y_c &\quad& \forall u \in P,\ \forall c \in F \nonumber \\
    &\sum_{c \in F} x_{u,c} \geq 1 &\quad& \forall u \in P \\
    & 0 \leq y_c \leq 1 &\quad& \forall c \in F \nonumber \\
    & 0 \leq x_{u,c} \leq 1 &\quad& \forall u \in P,\ \forall c \in F \nonumber
    \end{align}

    Note that this program is convex for every fixed $H$. We do not know the optimal $H$, so we try all values of $H$ separated by a multiplicative factor of $(1 + \epsilon)$, and then solve this relaxation for each $H$.

    We then show an algorithm that rounds the fractional solution of the above convex relaxation to an integer solution. 
    Let $(x^*,y^*,z^*, H^*)$ be the fractional solution of the above convex program. We use a rounding algorithm similar to that for the standard facility location problem by~\cite{shmoys1997approximation}. For each point $u \in X$, let $A_u = \sum_{c \in F} d(u,c)^p x^*_{u,c}$ denote the connection cost of point $u$ given by $(x^*,y^*, z^*, H^*)$. Then, we sort all points in $X$ by their connection costs in increasing order. We iteratively perform the following steps until all points are assigned to open facilities. We first pick the unassigned point $u$ with the smallest connection cost $A_u$ as an anchor point. Let $B_u = \{c \in F: d(u,c)^p \leq \alpha\cdot A_u\}$ be the set of all facilities in a ball with radius $(\alpha A_u)^{1/p}$ centered at $u$ for some constant $\alpha > 1$ specified later. Then, we open the facility $c_u$ in $B_u$ with the smallest opening cost. Assign all unassigned points $v$ with distance $d(v,u) \leq (\alpha A_u)^{1/p}+(\alpha A_v)^{1/p}$ to this new facility $c_u$. Let $\widehat{C}$ be the set of centers chosen by this rounding algorithm.

    We first bound the total cost for opening facilities. Since in each iteration, all unassigned points $v$ with distance $d(v,u) \leq (\alpha A_u)^{1/p}+(\alpha A_v)^{1/p}$ are assigned to the new facility, and $A_u$ is the smallest connection cost, the balls with radius $(\alpha A_u)^{1/p}$ centered at $u$ are disjoint across all iterations. 
    Thus, the sets $B_u$ are disjoint. 
    We then show that for any iteration with anchor point $u$
    \[
    \sum_{c \in B_u} x^*_{u,c} \geq \frac{\alpha-1}{\alpha}.
    \]
    Since $d(u,c)^p \geq \alpha A_u$ for $c \not \in B_u$, if $\sum_{c \in B_u} x^*_{u,c} < (\alpha - 1)/\alpha$, then $\sum_{c \not\in B_u} x^*_{u,c} > 1/\alpha$, which implies
    \[\sum_{c\in F} d(u,c)^p x^*_{u,c} = \sum_{c\in B_u} d(u,c)^p x^*_{u,c} + \sum_{c\not\in B_u} d(u,c)^p x^*_{u,c} > A_u,\] which contradicts the definition of $A_u$.
    Since $x^*_{u,c} \leq y^*_c$ for any $u$ and $c$, we have $\sum_{c\in B_u} y^*_c \geq \sum_{c\in B_u} x^*_{u,c} \geq (\alpha-1)/\alpha$.
    Thus, we have the opening cost for the facility $c_u$ chosen in this iteration is at most
    \[
    f(c_u) \leq f(c_u) \cdot \frac{\alpha}{\alpha-1} \sum_{c \in B_u} x^*_{u,c} \leq \frac{\alpha}{\alpha-1} \sum_{c \in B_u} x^*_{u,c} f(c) \leq \frac{\alpha}{\alpha-1} \sum_{c \in B_u} y^*_{c} f(c),
    \]
    where the second inequality is because the facility $c_u$ has the smallest opening cost in $B_u$. Since the sets $B_u$ are disjoint, we have the total opening cost given by the algorithm is at most 
    \[
    \sum_{c\in\widehat{C}} f(c) \leq \frac{\alpha}{\alpha-1} \cdot \sum_{c \in F} y^*_c \cdot f(c).
    \]

    We then bound the connection cost for each point $v \in X$. If $v$ is an anchor point chosen at some iteration, then this point $v$ is assigned to a center $c \in B_v$, which means $d(v,c)^p \leq \alpha \cdot A_v$. If $v$ is not an anchor point, then let $u$ be the anchor point of the iteration when $v$ is assigned to the facility $c_u$. By the triangle inequality, we have
    \[
    d(v,c) \leq d(u,c) + d(u,v) \leq (\alpha A_u)^{1/p} + (\alpha A_u)^{1/p}+(\alpha A_v)^{1/p} \leq 3 (\alpha A_v)^{1/p}.
    \]
    Thus, we have $d(v,c)^p \leq 3^p \alpha A_v$ for every $v \in X$. Then, for each group $i$, the connection cost for this group is at most
    \[
    \cost(w_i, \widehat{C})^{1/p} 
    = \left(\sum_{u\in X} w_i(u) \cdot d(u, \widehat{C})^p\right)^{1/p} \leq \left(3^p \alpha \sum_{u\in X} w_i(u) \cdot A_u \right)^{1/p} \leq 3 \alpha^{1/p} (z_i^*)^{1/q}
    \]
    Therefore, the total connection cost for this clustering is at most $3 \cdot \alpha^{1/p} \cdot H^*$.

    Combining the connection cost and the facility opening cost, the cost of the integral solution given by the rounding algorithm is at most 
    \[
    \max\left\{\frac{\alpha}{\alpha-1}, 3\cdot \alpha^{1/p}\right\} \cdot \mathrm{CP},
    \]
    where $\mathrm{CP}$ is the convex relaxation objective value given by $(x^*,y^*,z^*,H^*)$.
    By taking $\alpha = 4/3$, we have $\alpha/(\alpha-1) = 4$ and $3\cdot \alpha^{1/p} \leq 4$ for $p \geq 1$. Thus, by finding a $(1+\varepsilon/4)$ optimal solution for the convex relaxation, this rounding algorithm achieves a $4+\varepsilon$ approximation.
\end{proof}

We now provide the approximation algorithm for the group heterogeneous setting. 

\begin{proof}[Proof of Theorem~\ref{thm:fl_vector}]
    We use a convex relaxation of the problem and then use a randomized rounding procedure to get the integer solution. We consider the following convex relaxation
    \begin{align}
    \text{min} \quad & 
    H \nonumber\\
    \text{s.t.} \quad 
    & \sum_{i=1}^{\ell} z_i \leq H^q \nonumber \\
    & z_i \geq \left(\sum_{u} w_i(u) \sum_{c \in F} d(u, c)^p \cdot x_{u,c} \right)^{\frac{q}{p}} 
    &\quad& \forall i \in [\ell] \\
    & z_i \geq \left(\sum_{c\in F} y_c \cdot f_i(c)\right)^{q} &\quad& \forall i \in [\ell] \\
    & z_i \geq \sum_{c\in F} y_c \cdot f_i(c)^q
    &\quad& \forall i \in [\ell] \\
    &\sum_{c \in F} x_{u,c} \geq 1 &\quad& \forall u \in P\\
    & x_{u,c} \leq y_c &\quad& \forall u \in P,\ \forall c \in F \nonumber \\
    & 0 \leq y_c \leq 1 &\quad& \forall c \in F \nonumber \\
    & 0 \leq x_{u,c} \leq 1 &\quad& \forall u \in P,\ \forall c \in F \nonumber
    \end{align}

    We then apply the following randomized rounding algorithm inspired by~\cite{kasperski2015approximability,srinivasan1999approximation}. Let $(x^*,y^*, z^*, H^*)$ denote the fractional solution to the convex program above. 
    The rounding procedure proceeds iteratively, similar to the one used in the group homogeneous setting. 
    The key difference is that, instead of opening the cheapest facility in $B_u$, we open exactly one facility $c_u$ in $B_u$ selected randomly with probability proportional to $y^*_c$. 
    For each point $u \in X$, let $A_u = \sum_{c \in F} d(u,c)^p x^*_{u,c}$ denote the connection cost of point $u$ given by the fractional solution. Then, we sort all points in $X$ by their connection costs in increasing order. We iteratively perform the following steps until all points are assigned to open facilities. 
    We first pick the unassigned point $u$ with the smallest connection cost $A_u$ as an anchor point. Let $B_u = \{c \in F: d(u,c)^p \leq 2 A_u\}$ be all facilities in a ball with radius $(2 A_u)^{1/p}$ centered at $u$. 
    Then, we sample exactly one facility $c_u$ in $B_u$ with probability proportional to $y^*_c$.
    Assign all unassigned points $v$ with distance $d(v,u) \leq (2 A_u)^{1/p}+(2 A_v)^{1/p}$ to this new facility $c_u$. Let $\widehat{C}$ be the set of centers sampled by the rounding algorithm.

    We now analyze the cost of this integer solution given by the rounding algorithm. Let $\hat{y}_c \in \{0,1\}$ be the indicator variable to denote whether facility $c$ is selected in $\widehat{C}$. We first analyze the facility opening cost of this solution. With the analysis in Theorem~\ref{thm:fl_scalar}, we know that for each iteration with anchor node $u$, we have $\sum_{c \in B_u} y^*_c \geq 1/2$. Since we open exactly one facility $c_u$ in $B_u$ with probability proportional to $y^*_c$, we have the for each facility $c\in B_u$,
    \[
    \mathbb{E}[\hat{y}_c] = \frac{y^*_c}{\sum_{c\in B_u} y^*_c} \leq 2 y^*_c.
    \]
    Consider any fixed group $i$. For each iteration with anchor point $u$, we define a random variable $Y_u = \sum_{c\in B_u}\hat{y}_c f_i(c)$ which represents the contribution to the facility cost of group $i$ from that iteration. Since each facility is chosen independently in each iteration, the collection $\{Y_u\}$ forms a set of independent non-negative random variables. By the Rosenthal inequality (see ~\cite{johnson1985best}), we have 
    \[
    \mathbb{E}\left[\left(\sum_u Y_u\right)^q\right]^{1/q} \leq O\left(\frac{q}{\log q}\right) \max\left\{\left(\sum_u \mathbb{E}[Y_u^q]\right)^{1/q}, \sum_u \mathbb{E}[Y_u]\right\}.
    \]
    Note that in each iteration, exactly one facility $c_u$ in $B_u$ is opened, which means exactly one random variable $\hat{y}_{c_u} = 1$ and all other $\hat{y}_c = 0$ in $B_u$. Thus, we have 
    \[
    \mathbb{E}[Y_u^q] = \mathbb{E}\left[\left(\sum_{c \in B_u} \hat{y}_c\cdot f_i(c)\right)^q\right] = \mathbb{E}\left[\sum_{c \in B_u} \hat{y}_c\cdot f_i(c)^q\right] \leq 2 \sum_{c \in B_u} y^*_c \cdot f_i(c)^q.
    \]
    Hence, we have $\sum_u \mathbb{E}[Y_u^q] \leq 2 \sum_{c \in \bigcup_u B_u} y^*_c \cdot f_i(c)^q \leq 2 z^*_i$. We also have 
    \[
    \sum_u \mathbb{E}[Y_u] = \sum_{c \in \bigcup_u B_u} f_i(c) \mathbb{E}[\hat{y}_c] \leq 2 \sum_{c \in \bigcup_u B_u}f_i(c) y^*_c.
    \]
    Therefore, we have 
    \[
    \mathbb{E}\left[\left(\sum_{c \in \bigcup_u B_u} \hat{y}_c f_i(c)\right)^q\right] \leq O\left(\left(\frac{2q}{\log q}\right)^q\right) z^*_i.
    \]
    By taking the sum over all groups and applying Markov's inequality, we have with probability at least $1/2$,
    \[
    \sum_{i=1}^\ell \left(\sum_{c \in \bigcup_u B_u} \hat{y}_c f_i(c)\right)^q \leq O\left(\left(\frac{2q}{\log q}\right)^q\right) \sum_{i=1}^\ell z^*_i \leq O\left(\left(\frac{2q}{\log q}\right)^q\right) (H^*)^q
    \]

    We then analyze the connection cost of this integer solution. Since the algorithm picks exactly one center from each ball $B_u$, by the analysis in Theorem~\ref{thm:fl_scalar}, we have the connection cost for each group $i$ given by $\widehat{C}$ is at most 
    \[
    \cost(w_i,\widehat{C})^{1/p} 
    = \left(\sum_{u\in X} w_i(u) \cdot d(u, \widehat{C})^p\right)^{1/p} \leq \left(3^p \cdot 2 \sum_{u\in X} w_i(u) \cdot A_u \right)^{1/p} \leq 6 (z_i^*)^{1/q}.
    \]
    By combining with the facility opening cost, we have the facility location cost is at most
    \begin{align*}
    &\left(\sum_{i=1}^\ell \left(\cost(w_i,\widehat{C})^{1/p} + \sum_{c \in \widehat{C}} f_i(c)\right)^q \right)^{1/q} 
    \\
    \leq&  \left(\sum_{i=1}^\ell 2^q\cost(w_i, \widehat{C})^{q/p} + 2^q\left(\sum_{c \in \widehat{C}} f_i(c) \right)^q \right)^{1/q} 
    \leq O\left(\frac{q}{\log q}\right) \cdot H^*.
    \end{align*}
    Therefore, by computing a $(1+\varepsilon)$ approximate solution to the convex relaxation, this randomized rounding provides the desired approximation with probability at least $1/2$. 
\end{proof}

We then present hardness results for the Socially Fair Facility Location problem with group-heterogeneous facility costs and $\ell_1$-Clustering cost.
Specifically, we show that this problem can not be approximated within $O(\log^{1-\varepsilon} \ell)$ for any $\varepsilon > 0$ unless $NP \subseteq DTIME(n^{\mathrm{poly}(\log n)})$. 
For this problem, by replacing the convex relaxation used in Theorem~\ref{thm:fl_vector} by a linear programming relaxation, we achieve the same approximation factor of $O(\log\ell / \log\log \ell)$. 
We further show that this LP relaxation has an integrality gap of $\Omega(\log\ell / \log\log \ell)$. 
Both results are established via a reduction from the min-max representation selection problem studied in~\cite{kasperski2015approximability}. In the min-max representation selection problem, we are given $n$ tools partitioned into $k$ disjoint subsets, where each tool is associated with a cost vector in $\mathbb{R}^\ell$. The goal is to select one tool from each subset to minimize the $\ell_\infty$ norm of the sum of the selected tools’ cost vectors.

\begin{theorem}\label{thm:fl_vector_hardness}
    For the Socially Fair Facility Location with group heterogeneous facility cost and $\ell_1$-Clustering cost, there is no $O(\log^{1-\varepsilon} \ell)$ approximation algorithm for any $\varepsilon > 0$ unless $NP \subseteq DTIME(n^{\mathrm{poly}(\log n)})$.
    Moreover, the linear programming relaxation has an $\Omega(\log\ell / \log\log \ell)$ integrality gap.
\end{theorem}

\begin{proof}
    We consider the following reduction. 
    Consider any instance of the min-max representation selection problem with $n$ tools, $k$ disjoint subsets $T_1,\cdots, T_k$, and cost vector $f(i) \in \mathbb{R}^\ell$ for each tool $i$.
    Without loss of generality, we assume that the cost vector for each tool is bounded in $[0,1]^\ell$. Otherwise, we can scale the cost vector to satisfy this.  
    We construct an instance for the Socially Fair Facility Location problem as follows. 
    We define a metric space consisting of $k$ points $u_1,\cdots, u_k$ such that every pair of points is placed far apart, $d(u_i,u_j) \geq n\ell$. 
    The client set $X$ consists of $n$ points with $|T_i|$ clients located at $u_i$ for each $i \in [k]$. 
    Then, each client corresponds to a tool.
    There are $\ell$ groups, each assigning weight $w_i(u) = 1$ to every client $u$.
    For each client point in $X$, we place one facility at the same location. The facility opening cost is set to match the cost vector associated with the corresponding tool.

    We now argue the correctness of the reduction.
    First, observe that any solution achieving better than an $\ell$ approximation must open exactly one facility at each location $u_1,\cdots, u_k$.
    Suppose, for contradiction, that no facility is opened at some point $u_i$.
    Then the clients at $u_i$ would be assigned to facilities at a distance of at least $n\ell$, resulting in a connection cost of at least $n\ell$ for every group.
    In contrast, any solution which opens one facility at each $u_i$ incurs a total cost of at most $n$ - the facility cost is at most $n$ for each group, and the connection cost is $0$.
    Therefore, this solution incurs a total cost worse than an $\ell$ approximation.
    Hence, any solution with an approximation factor better than $\ell$ must open one facility at each point $u_i$, leading to zero connection cost for every group.
    In this case, the total cost equals the facility opening cost, which exactly matches the objective value of the corresponding min-max representation selection instance.
\end{proof}

\bibliographystyle{plainnat}
\bibliography{references}

\clearpage
\appendix

 \section{\texorpdfstring{$\ell_1$}{l-1}-Clustering: 3-approximation with 
 \texorpdfstring{$k + f(\ell)$}{additional} centers}
 \label{sec:extra-centers}
 In this section, we show a constant approximation for $\ell_1$-Clustering with $k+f(\ell)$ centers using local search.

 \begin{theorem}\label{thm:localsearchconstant}
 Given an instance of $\ell_1$-Clustering, for every constant $0 < \varepsilon' < \frac{1}{10}$ there is an algorithm that runs in $g(\ell,\frac{1}{\varepsilon'})n^{O\left(\frac{1}{\varepsilon'}\right)}$ time, opens $k + f(\ell, \frac{1}{\varepsilon'})$ centers and gives a $(3 + 2\varepsilon')$ approximation. 
  \end{theorem}

 As in our previous local search algorithm in~Section~\ref{sec:local-search}, we maintain a solution $C$ throughout the course of the algorithm. Given the instance of $\ell_1$-clustering, let $z_i = \cost(w_i) = \sum_{u \in X} w_i(u) d(u,C)$ for $i \in [\ell]$. Let $z_i^* = \cost^*(w_i) = \sum_{u \in X} w_i(u) d(u,C^*)$, where $C^*$ is an optimal solution, and $z_i^*$ is the optimal cost for weight $i \in [\ell]$.

 Since we are dealing with an instance of $\ell_1$-clustering, without loss of generality, let us fix $\OPT = \max_{i \in [\ell]} z_i^* = 1$ by scaling distances. First, we start with the definition of $t$-swaps.

\begin{definition}
Given a set $C$ of current centers, a $t$-swap is a set $S$ satisfying $|S \setminus C| = |S \cap C| \leq t$ which produces a new set $C' = (C \setminus S) \cup (S \setminus  C)$. In other words, at most $t$ centers of $C$ are swapped out for centers outside. Note that after the $t$-swap, the number of centers in $C$ remains unchanged.

\end{definition}

 We start with a simple lemma about $t$-swap local search closely following our analysis in~\Cref{thm:basic-local-search}. The difference is that \Cref{thm:basic-local-search} relies on single swaps ($t = 1$), while we use $t$-swaps with $t > 1$ in this section to obtain tighter bounds, and in particular, the approximation ratio of $3$. Let $\cost_S(w)$ denote the new cost with respect to weights $w$ after applying a $t$-swap $S$.

 \begin{lemma}\label{lemma:basic-local-search-multiswap}

Let \((X,d)\) be a metric space. Given weights \(w\), the sets \(C,C^*\) as defined above,
parameters \(M\ge m\ge 0\), and coefficients \(\xi_S\in[m,M]\), there exists a probability distribution
\(\mathcal D\) over \(t\)-swaps \(S\) such that
\[
\mathbb E_{S\sim\mathcal D}\!\left[\xi_S\cdot(\cost_S(w)-\cost(w))\right]
\le
\frac{3M(1+1/t)\cost^*(w)-m\cost(w)}{k}.
\] 
 \end{lemma}
 \begin{proof}

 The lemma follows from an easy application of the following lemma which follows directly from~\cite[Lemma~2.7]{GT}.

 \begin{lemma}\label{lemma:multiswap}
 Given a set of centers $C$, a corresponding optimal assignment $\beta: X \rightarrow C$, an optimal set of centers $C^*$ with corresponding optimal assignment $\beta^*: X \rightarrow C^*$ such that $|C| \geq |C^*| = k$, there is a distribution $\mathcal{D}$ over $t$-swaps $S$, such that the following hold.
 \begin{enumerate}
 \item $S \subseteq C \cup C^*$.
 \item $\Pr_{S \sim \mathcal{D}}[c \in S] \leq \frac{1 + 1/t}{k}$ for every $c \in C$.
 \item $\Pr_{S \sim \mathcal{D}}[c^* \in S] = \frac{1}{k}$ for every $c^* \in C^*$.
 \item $\cost_S(w) - \cost(w) \leq \sum_{u: \beta(u) \in S \cap C} 2w(u)x_u^* - \sum_{u: \beta^*(u) \in S \cap C^*} w(u) (x_u - x_u^*)$ for every $S \in \mathrm{supp}(\mathcal{D})$, where $x_u^* = d(u,C^*)$ and $x_u = d(u,C)$ for each $u \in X$.

 \end{enumerate}

 \end{lemma}

 We use the same distribution $\mathcal{D}$ as given by this lemma. Now we have

 \begin{align*}
\Exp_{S\sim\mathcal{D}}[\xi_{S}(\cost_S(w)-\cost(w))]
&\leq M\sum_{S:\,\cost_S(w)\ge\cost(w)}
   \Pr[S]\bigl[\cost_S(w)-\cost(w)\bigr]\\
&\quad +m\sum_{S:\,\cost_S(w)\le\cost(w)}
   \Pr[S]\bigl[\cost_S(w)-\cost(w)\bigr]\\
&\leq \sum_{u\in X}\frac{3M(1+1/t)}{k}\,w(u)x_u^*
   \;-\;\sum_{u\in X}\frac{m}{k}\,w(u)x_u\\
\end{align*}
where in the last inequality we use all properties in~\Cref{lemma:multiswap}, together with the fact that for each $u \in X$, $\beta(u)$ and $\beta^*(u)$ map to exactly one center of $C$ and $C^*$ respectively, and multiply all positive terms by $M$, and negative terms by $m$. The right hand side is then at most
$$\frac{3M(1+1/t)\,\cost^*(w)\;-\;m\,\cost(w)}{k}$$
as desired. This concludes the proof.
\end{proof}

As in~Section~\ref{sec:local-search}, we define a potential function:
 \[\Phi = \sum_{i \in [\ell]} z_i^q,\]
 where $q \approx \log \ell$ will be determined from the analysis.

Fix a $t$-swap $S$, a parameter $\varepsilon$ to be determined later, and suppose the cost of weight $i$ changes from $z_i$ to $z_i'$. For a real number \(a\), write \(a_+ = \max\{a,0\}\).

We define the linearized potential change

\[
\Delta(S) =
\sum_{i:z_i\ge z'_i} q(z_i-\varepsilon)_+^{q-1}(z'_i-z_i)
+
\sum_{i:z_i\le z'_i} q(z_i+\varepsilon)^{q-1}(z'_i-z_i).
\]
Intuitively, $\Delta(S)$ upper bounds the change in $\Phi$ for swap $S$ assuming that $z_i' \in [z_i - \varepsilon, z_i + \varepsilon]$. Note that $\Delta(S)$ linearly depends on $z_i'-z_i$. The next lemma shows that if our current solution is not a $3$-approximation, then there must exist a distribution over swaps such that the linearized potential change drops in expectation (and in particular, that there is some swap $S$ for which $\Delta(S)$ < 0). 

 \begin{lemma}\label{lemma:lpcreduce}
 For any $0 < \varepsilon' < \frac{1}{10}$ and for $t \geq \frac{1000}{\varepsilon'}$, $q \geq 1 + \frac{100}{\varepsilon'^2}{\log \ell}$, $\varepsilon = \frac{1}{20tq}$ if there exists $i \in [\ell]$ such that $z_i \geq 3(1 + 2\varepsilon')$, then there exists a distribution $\mathcal{D}$ over $t$-swaps $S$, such that $\mathbb{E}_{S\sim\mathcal{D}}[\Delta(S)] \leq \frac{-\varepsilon^q}{k}$.
 \end{lemma}

 \begin{proof}
 We consider the distribution $\mathcal{D}$ output by~\Cref{lemma:basic-local-search-multiswap} and show that it satisfies the condition of the lemma. Fix an $i \in [\ell]$. We let $m_i = q(z_i-\varepsilon)_+^{q-1}$, and $M_i = q(z_i + \varepsilon)^{q-1}$. 
For every $t$-swap $S$, define 

$$
\xi_{S,i} = 
\begin{cases}
m_i, & \text{if } z_i' < z_i\\
M_i, & \text{otherwise}
\end{cases}
$$

Then~\Cref{lemma:basic-local-search-multiswap} guarantees that
\[
\mathbb E_{S\sim\mathcal D}[\xi_{S,i}(z'_i-z_i)]
\le
\frac{3(1+\delta)M_i z_i^* - m_i z_i}{k},
\]
where \(\delta=1/t\).


Notice that since $t \geq 1000/\epsilon'$, we have $\delta \leq \epsilon'/1000$. Further note that $\sum_{i} \mathbb{E}_{S}[\xi_{S,i}(z_i' - z_i)] = \mathbb{E}_{S}[\Delta(S)]$.

Observe that for any $i \in [\ell]$ with $z_i \geq 3$, we have
\begin{equation}\label{eq:boundratio}
M_i/m_i = (z_i + \varepsilon)^{q-1} /(z_i - \varepsilon)^{q-1} \leq (1 + 3\varepsilon/z_i)^q \leq 1 + \delta/10
\end{equation}
where we use the inequality $(1 + v) \leq e^v $ and $e^u \leq 1 + 3u$ for $0 \leq u,v \leq 1$ and since we set $\varepsilon = \frac{\delta}{20q}$.

Now we divide the indices $i \in [\ell]$ into two groups, depending on if their current cost is large or small.

\begin{enumerate}
\item  Large $i$, satisfying $z_i \geq 3(1 + 5\delta)$.
In this case, note that the contribution of this term to $k\mathbb{E}_S[\Delta(S)]$ is

\begin{align*}
3M_i(1 + \delta)z_i^* - m_i z_i &= m_i\bigl(\frac{3M_i}{m_i
}(1 + \delta)z_i^* - z_i\bigr) \\
  &\leq m_i\bigl(3(1 + \delta/10)(1 + \delta) - z_i\bigr) \;\;\text{(using~\Cref{eq:boundratio} and $z_i^* \leq 1$)}\\ &\leq -3\delta m_i \\  
  &\leq 0
\end{align*}

\item Small $i$, satisfying $z_i \leq  3(1 + 5\delta)$. In this case, note that the contribution of this term to $k\mathbb{E}_S[\Delta(S)]$ is at most $3(1 + \delta)M_iz_i^* \leq 3q(1 + \delta)(3(1 + 5\delta) + \varepsilon)^{q-1}$.
\end{enumerate}

Thus the total contribution of all the small $i$ is at most $$\lambda_{\text{small}} = 3q\ell(1 + \delta)(3(1 + 5\delta + \varepsilon))^{q-1}.$$

We wish to show that if there is an $i \in [\ell]$ satisfying $z_i \geq 3(1 + 2\varepsilon')$, then $k\mathbb{E}_{S}[\Delta(S)] < 0$.

If there is an $i$ with $z_i \geq 3\ell^{\frac{1}{q-1}}(1 +\varepsilon')(1 + 5\delta + \varepsilon)$, then its contribution to $k\Exp_S[\Delta(S)]$ is at most

\begin{align*}
3M_i(1 + \delta)z_i^* - m_i z_i &= m_i\bigl(\frac{3M_i}{m_i}(1 + \delta)z_i^* - z_i\bigr) \\
  &\leq m_i(3(1 + \delta/10)(1 + \delta) - z_i) \;\;\text{(using~\Cref{eq:boundratio} and $z_i^* \leq 1$)}\\ &\leq -3\varepsilon'm_i \\
  &\leq -3\varepsilon'q(3\ell^{1/(q-1)}(1 + \varepsilon')(1 + 5\delta + \varepsilon) - \varepsilon)^{q-1} \\
  &\leq -{3q\ell \varepsilon'} \bigl(3(1 + \varepsilon'/2)(1 + 5\delta + \varepsilon)\bigr)^{q-1}\;\;\text{(since $\varepsilon = \frac{\delta}{20q} \leq \delta \leq \frac{\varepsilon'}{1000}$)}
\end{align*}

Thus if

\begin{equation}\label{eqn:decrease}
{3q\ell \varepsilon'} \bigl(3(1 + \varepsilon'/2)(1 + 5\delta + \varepsilon)\bigr)^{q-1} \geq 2\lambda_{\text{small}} = 6q\ell(1 + \delta)(3(1 + 5\delta + \varepsilon))^{q-1}
\end{equation}

we must have $k\Exp_{S}[\Delta(S)] \leq -{1.5q\ell \varepsilon'} \bigl(3(1 + \varepsilon'/2)(1 + 5\delta + \varepsilon)\bigr)^{q-1} \leq -\varepsilon^q$ as desired, where we again use the fact $\epsilon = \frac{\delta}{20q} \leq \delta \leq \epsilon'/1000$. The inequality~\ref{eqn:decrease} is satisfied if

\begin{align*}
{\varepsilon'} &\geq 2(1 + \delta)\Bigl(\frac{1}{1 + \varepsilon'/2}\Bigr)^{q-1}\\
\frac{2}{\varepsilon'}(1 + \delta) &\leq (1 + \varepsilon'/2)^{q-1}
\end{align*}

which is satisfied since $q \geq \frac{100}{\varepsilon'^2}$.

Finally, we show that $3\ell^{\frac{1}{{q-1}}}(1 +\varepsilon')(1 + 5\delta + \varepsilon) \leq 3(1 + 2\varepsilon')$.

We have $\delta \leq \varepsilon'/100$, $\varepsilon \leq \frac{\varepsilon'}{100q}$. Since $q - 1\geq \frac{100}{\varepsilon'^2}\log {\ell}$, we must have $\ell^{\frac{1}{q-1}} \leq 2^{\frac{\varepsilon'^2}{100}} \leq (1 + \frac{\varepsilon'}{100})$. It therefore follows that
\begin{align*}
3\ell^{\frac{1}{{q-1}}}(1 +\varepsilon')(1 + 5\delta + \varepsilon) &\leq 3(1 + \frac{\varepsilon'}{100}) (1 + \varepsilon')(1 + \frac{\varepsilon'}{10}) \;\;\text{(since $\varepsilon = \frac{\delta}{20q} \leq \delta \leq \frac{\varepsilon'}{1000}$)}
\\&\leq 3(1 + 2\varepsilon').
\end{align*}

Thus we have shown that if there exists even one $i \in [\ell]$ with $z_i \geq 3(1 + 2\varepsilon')$, we must have $\mathbb{E}_S[\Delta(S)] \leq -\frac{\varepsilon^q}{k}$, which finishes the proof.

\end{proof}

Henceforth, we will fix $\varepsilon'$, $t =  \frac{1000}{\varepsilon'}$, $q = 1 +  \frac{100}{\varepsilon'^2}{\log \ell}$ and $\varepsilon = \frac{1}{20tq}$ .

The next lemma shows that if there is a swap $S$ for which the cost decreases significantly for a group, then simply opening the centers in $S \setminus C$ \emph{without} closing the centers in $C \cap S$ must reduce the potential significantly.

\begin{lemma}\label{lemma:decrease}
Let $S$ be a $t$-swap such that $z_i' \leq z_i - \varepsilon$ for some group $i \in [\ell]$. Then opening the centers in $S \setminus C$ (without closing any centers of $C$) reduces the potential function by at least $\varepsilon^q$.
\end{lemma}

\begin{proof}
Let $z_i''$ be the cost of group $i$ in the solution obtained by opening $S$ as additional centers. Clearly, $z_i'' \leq z_i - \varepsilon$ by our assumption, and since we do not close any existing centers, we must have $z_j'' \leq z_j$ for each group $j \in [\ell]$.

But then $z_i''^q - z_i^q \leq (z_i - \varepsilon)^q - z_i^q \leq -\varepsilon^q$ since $q \geq 1$, and this finishes the proof.
\end{proof}

The following lemma shows that if there is a swap $S$ for which the cost \emph{increases} significantly for a group and additionally $\Delta(S) < 0$ for the swap $S$,
then simply opening the centers in $S$ \emph{without} closing the centers in $C \cap S$ must reduce the potential significantly.

\begin{lemma}\label{lemma:increase}
Let $S$ be a $t$-swap such that $z_i' \geq z_i + \varepsilon$ for some group $i \in [\ell]$, and $\Delta(S) < 0$. Then opening the centers in $S \setminus C$ (without closing any centers of $C$) reduces the potential function by at least $\varepsilon^q$.
\end{lemma}

\begin{proof}
First, if $z_j' \leq z_j - \varepsilon$ for some $j \in [\ell]$, by the previous lemma, we are already done.

Therefore we can assume that $z_j' \geq z_j - \varepsilon$ for each group $j \in [\ell]$. Since $\Delta(S) < 0$, we have

$$\sum_{j: z_j \geq z_j'} q(z_j - \varepsilon)_+^{q-1} (z_j' - z_j) + \sum_{j:z_j \leq z_j'} q(z_j + \varepsilon)^{q-1}(z_j'- z_j) < 0$$

In particular, this means that
\begin{equation}\label{eq:lpc}
\sum_{j:z_j \geq z_j'} q(z_j - \varepsilon)_+^{q-1} (z_j' - z_j) \leq - q(z_i + \varepsilon)^{q-1}(z_i' -z_i) \leq -q\varepsilon^{q}.
\end{equation}

Now consider what happens when we open $S \setminus C$. If the cost of each group $j$ after opening these centers is $z_j''$ and the final potential is $\Phi''$, then we have:

\begin{align*}
 \Phi'' - \Phi &\leq \sum_{j: z_j'' < z_j} (z_j''^q - z_j^q) \\
 &\leq \sum_{j: z_j' \leq  z_j} (z_j'^q - z_j^q) \\
 &\leq \sum_{j: z_j' \leq  z_j} q(z_j - \varepsilon)_+^{q-1} (z_j' - z_j)\;\; \intertext{where the last inequality follows from the Mean value theorem and the fact that $z_j' \geq z_j - \varepsilon$ for each $j$.}
 \end{align*}
 Finally, by~\Cref{eq:lpc}, we have 
 \begin{align*}
 \sum_{j: z_j' \leq z_j} q(z_j - \varepsilon)_+^{q-1}(z_j'-z_j)&\leq -q(z_i + \varepsilon)^{q-1}(z_i'-z_i) \\ 
 &\leq -q\varepsilon^q
\end{align*}
as desired.
\end{proof}

\begin{proof}[Proof of~\Cref{thm:localsearchconstant}]
First, run the algorithm of~\Cref{thm:local_search} with $p = 1$ and $q = \log \ell$. The $\ell_q$ norm then gives a constant approximation to the $\ell_{\infty}$ norm, and hence we obtain an $O(\log \ell)$ approximation for $\ell_1$-clustering with $k$ centers. Thus at this point we have a set $C$ of $k$ centers, and the cost $z_i$ of each group $i \in [\ell]$ satisfies $z_i \leq O(\log \ell)$. Choose $\varepsilon' > 0$, redefine $q = \frac{100}{\varepsilon'^2} \log \ell$ and choose $t =  \frac{1000}{\varepsilon'}$, $\varepsilon = \frac{1}{20tq}$ as in the above analysis.

If $z_i \leq 3(1 + 2\varepsilon')$ for every $i \in [\ell]$, stop and output the currrent solution $C$.  Else, by~\Cref{lemma:lpcreduce}, there must exist a $t$-swap $S$ such that $\Delta(S) \leq \frac{-\epsilon^q}{k}$. By trying all $t$-swaps, find a $t$-swap $S$ such that $\Delta(S)$ is minimized. If after applying this swap, there is some group $i$ such that $|z_i' - z_i| > \varepsilon$, then open the set of centers in $S \setminus C$ without closing any center of $C$ to obtain a new set of centers $C'$. Else, execute the swap $S$ to obtain a new set of centers $C' = (C \setminus S) \cup (S \setminus C)$. We now repeat the procedure with $C$ replaced by $C'$. This concludes the description of the algorithm.

The next lemma shows that throughout the course of the algorithm, $\Phi$ decreases in every step.

\begin{lemma}

After every iteration of the algorithm, $\Phi$ decreases by at least $\frac{\varepsilon^q}{k}$.
\end{lemma}

\begin{proof}

If we open the set of centers $S \setminus C$, it must be the case that $|z_i - z_i'| \geq \epsilon$ for some group $i$. In this case, by~\Cref{lemma:decrease,lemma:increase}, $\Phi$ decreases by at least $\varepsilon^q$.
Otherwise, for the swap $S$, for every group $j$, we must have $|z_j - z_j'| \leq \varepsilon$, and $\Delta(S) \leq \frac{-\varepsilon^q}{k}$.

Now applying the mean-value theorem, we have

\begin{align*}
z_j'^q - z_j^q = q\theta^{q-1}(z_j' - z_j) 
\end{align*}

for some $\theta$ between $z_j$ and $z_j'$. Since $|z_j' - z_j| \leq \varepsilon$, we obtain the two inequalities

\[
z_j'^q - z_j^q \le
\begin{cases}
q(z_j + \varepsilon)^{q-1}(z_j' - z_j), & \text{if } z_j' \geq z_j,\\
q(z_j-\varepsilon)_+^{q-1}(z_j' - z_j), & \text{if } z_j' \leq z_j.
\end{cases}
\]

Let $\Phi'$ be the new potential after applying the swap $S$. This implies

\begin{align*}
\Phi' - \Phi &= \sum_{j}(z_j'^q - z_j^q) \\
&\leq \sum_{j: z_j' \geq z_j}q(z_j + \varepsilon)^{q-1}(z_j' - z_j) + \sum_{j: z_j' \leq z_j}q(z_j-\varepsilon)_+^{q-1}(z_j' - z_j) \\
&= \Delta(S) \leq \frac{-\varepsilon^q}{k}.
\end{align*}

\end{proof}

\paragraph{Bound on the number of additional centers.}
Note that since we start with a $O(\log \ell)$ approximation, $z_i \leq O(\log \ell)$ for each $i \in [\ell]$ initially. Thus the initial potential $\Phi_0 = \sum_{i \in [\ell]} z_i^q \leq O(\ell \log^q \ell)$. Throughout the algorithm $\Phi$ is non-increasing, and whenever we open at most $t$ additional centers (without swapping), by~\Cref{lemma:decrease,lemma:increase} we decrease the potential by at least $\varepsilon^q$. Thus the total number of additional centers is at most $t\phi_0 \leq f(\ell, \frac{1}{\varepsilon'}) = O\bigl(\frac{\ell}{\epsilon'} (\frac{\log \ell}{\varepsilon})^q\bigr)$.

\paragraph{Bound on the number of iterations and running time.} Every time we make a swap $S$, the potential reduces by at least $-\Delta(S) > \frac{1}{k}\varepsilon^{q}$. Everytime we open a center without closing existing centers, the potential reduces by at least $\epsilon^q$. Thus the number of iterations is at most $O(k\ell(\frac{\log \ell}{\varepsilon})^q)$. In each iteration the algorithm tries all possible $t$-swaps, taking $n^{O(t)}$ time. Since $t = O(\frac{1}{\varepsilon'})$, and the number of iterations is at most $kg(\ell, \frac{1}{\varepsilon'})$ for some function $g$, the running time of the algorithm is bounded by $g(\ell, \frac{1}{\varepsilon'})n^{O\left(\frac{1}{\varepsilon'}\right)}$.

\end{proof}

\end{document}